\documentclass[12pt]{article}
\usepackage[utf8]{inputenc}
\usepackage[margin=1in]{geometry}
\usepackage[super,numbers,compress]{natbib}
\usepackage{amsmath,mathtools,graphicx,bm,amsthm,amsfonts,booktabs,url,appendix}
\DeclarePairedDelimiter\abs{\lvert}{\rvert}
\DeclarePairedDelimiter\norm{\lVert}{\rVert}

\DeclareMathOperator*{\argmin}{arg\,min}
\newcommand\indep{\protect\mathpalette{\protect\independenT}{\perp}}
\def\independenT#1#2{\mathrel{\rlap{$#1#2$}\mkern2mu{#1#2}}}
\newcommand\V{\mbox{Var}}
\newcommand\Cov{\mbox{Cov}}
\newtheorem{theorem}{Theorem}
\newtheorem{corollary}{Corollary}
\newtheorem*{remark}{Remark}

\usepackage{setspace}

\title{Novel Methods for the Analysis of Stepped Wedge Cluster Randomized Trials\thanks{Please address correspondence to: Lee Kennedy-Shaffer, Department of Biostatistics, Harvard T.H. Chan School of Public Health, 655 Huntington Ave., Building 2, 4th Floor, Boston, MA, 02115, USA. Email: lee\_kennedyshaffer@g.harvard.edu. This work was supported by the National Institute of Allergy and Infectious Diseases Award Numbers 5T32AI007358-28 and 1F31AI147745 (for L.K.S.), and R3751164 (for V.d.G.), and by the National Institute of General Medical Sciences Award Number U54GM088558 (for M.L.). The authors wish to thank Professor Michael D. Hughes for his valuable comments and feedback on the research at various stages. The authors also wish to thank Professor Anete Trajman for making the data available from the stepped wedge cluster randomized trial described in Section \ref{Example}.}}
\author{Lee Kennedy-Shaffer\thanks{Department of Biostatistics, Harvard T.H. Chan School of Public Health, Boston, MA, USA}, Victor De Gruttola\footnotemark[2], Marc Lipsitch\thanks{Department of Epidemiology, Department of Immunology and Infectious Diseases, and Center for Communicable Disease Dynamics, Harvard T.H. Chan School of Public Health, Boston, MA, USA}}
\date{September 16, 2019}

\begin{document}

\maketitle

\section*{Abstract}

Stepped wedge cluster randomized trials (SW-CRTs) have become increasingly popular and are used for a variety of interventions and outcomes, often chosen for their feasibility advantages.  SW-CRTs must account for time trends in the outcome because of the staggered rollout of the intervention inherent in the design. Robust inference procedures and non-parametric analysis methods have recently been proposed to handle such trends without requiring strong parametric modeling assumptions, but these are less powerful than model-based approaches. We propose several novel analysis methods that reduce reliance on modeling assumptions while preserving some of the increased power provided by the use of mixed effects models. In one method, we use the synthetic control approach to find the best matching clusters for a given intervention cluster. This approach can improve the power of the analysis but is fully non-parametric. Another method makes use of within-cluster crossover information to construct an overall estimator. We also consider methods that combine these approaches to further improve power. We test these methods on simulated SW-CRTs and identify settings for which these methods gain robustness to model misspecification while retaining some of the power advantages of mixed effects models. Finally, we propose avenues for future research on the use of these methods; motivation for such research arises from their flexibility, which allows the identification of specific causal contrasts of interest, their robustness, and the potential for incorporating covariates to further increase power. Investigators conducting SW-CRTs might well consider such methods when common modeling assumptions may not hold.
\medskip \newline
\textbf{Keywords:} Stepped wedge, cluster randomized trials, mixed effects models, permutation tests, synthetic control

\section{Background}

Cluster randomized trials (CRTs) have become a popular form of randomized trial, with many practical benefits, reflecting the necessity of implementing some interventions on clusters of individuals, and statistical benefits, such as accounting for interference between individuals.\cite{Halloran2010-zd,Eldridge2012-hg,Hayes2017} The causal estimand of interest and the overall risk-benefit profile of the trial  can also affect the choice to use cluster randomization.\cite{Kahn2018-ku,Hitchings2018-is,Bellan2019-dk} While parallel-arm CRTs are the most common, stepped wedge CRTs (SW-CRTs) have also become more common, being used for a variety of interventions.\cite{Brown2006-wu,Hemming2015-yq,Beard2015-kp,Davey2015-no,Barker2016-gn} In SW-CRTs, each cluster begins in the control arm. At designated time points, a cluster or clusters ``cross over'' to the intervention arm and remain in that arm for the duration of the study. The order in which clusters cross over to the intervention is randomized.\cite{Hussey2007-ay,Copas2015-so}

SW-CRTs are especially valuable when the intervention cannot be implemented in a large number of clusters simultaneously due to practical constraints.\cite{Kahn2018-ku,Brown2006-wu,Hemming2015-yq,Prost2015-ug} They can also be useful when the communities who will participate in the trial wish to ensure that all clusters receive the intervention before the end of the trial.\cite{Brown2006-wu,Hemming2015-yq,Prost2015-ug,Tugwell2019-nn} In particular, the design can be useful for assessing complex health interventions and for evaluating effectiveness of implementation.\cite{Mdege2011-mc,Keriel-Gascou2014-yh} There are, however, drawbacks to the design, and some of the benefits of the design may be achieved with parallel-arm CRT designs as well.\cite{Kotz2012-rc} Ethical arguments both for and against SW-CRTs have been made in various contexts, including arguments about the role of clinical equipoise.\cite{Brown2006-wu,Hargreaves2015-zp,Eyal2017-gf} And while the design may yield increased power over parallel-arm CRTs, this depends on both a large number of measurements over time and a statistically valid analysis method that controls for confounding of the treatment effect by time.\cite{Copas2015-so,Keriel-Gascou2014-yh,Kotz2012-rc,Hargreaves2015-zp,Mdege2014-pm,Viechtbauer2014-ox} Because more clusters are assigned to the intervention at later time points, the effect of time on the outcome must be accounted for in order to obtain unbiased or consistent treatment effects.\cite{Davey2015-no,Hargreaves2015-zp,Nickless2018-cm} Additionally, SW-CRTs with a relatively small number of clusters can be underpowered to detect effects, at least without making strong modeling assumptions.\cite{Thompson2017-or}

The most common method for analyzing SW-CRTs is the use of a linear or generalized linear mixed effects model. As described by Hussey and Hughes, this model can include a random intercept for each cluster and a fixed effect for time periods.\cite{Hussey2007-ay} This form of the model assumes that the additive effect for each time period is the same across clusters. A more general model proposed by Hooper et al.\ adds an independent random intercept for cluster-period.\cite{Hooper2016-ih} This, however, still assumes that the time trend does not vary systematically among clusters. In addition, both models require the specification of the distribution of these random intercepts. Misspecified random effects distributions can affect inference on the fixed effect estimators (i.e., the treatment effect estimator), although the effect on fixed effect estimates themselves is unclear and context-dependent.\cite{Thompson2017-or,Hartford2000-zz,Heagerty2001-vg,Agresti2004-qa,Litiere2007-an,Litiere2008-gw,McCulloch2011-jz} Finally, for the relatively small number of clusters in many SW-CRTs, asymptotic inference may not be appropriate.\cite{Wang2017-gr,Ji2017-cq}

Various methods have been proposed to remedy these issues. One approach, proposed by Wang and De Gruttola and by Ji et al., uses permutation tests to ensure nominal Type I Error and accurate inference, even for small numbers of clusters, as long as the effect estimate is unbiased.\cite{Wang2017-gr,Ji2017-cq} In the longitudinal context more generally, linear and generalized linear mixed effects models have been proposed that allow for flexible semi-parametric specification of the random effects distributions.\cite{Davidian1993-gb,Zhang2001-gg,Chen2002-zy} The operating characteristics of these different approaches to robust mixed effect model fitting have not been well-studied for SW-CRTs. Scott et al.\ have proposed the use of generalized estimating equations with finite-sample corrections to avoid the need to specify random effects distributions.\cite{Scott2017-wi} Thompson et al.\ recently proposed a non-parametric analysis method that uses within-period (``vertical'') comparisons.\cite{Thompson2018-so} They propose conducting inference by permutation tests as well to ensure nominal Type I Error and confidence interval coverage. They demonstrate through simulation that this method has no or low bias and nominal Type I Error and coverage.\cite{Thompson2018-so} Finally, Hughes et al.\ have proposed a robust inference method for SW-CRTs using vertical comparisons that gives a closed-form standard deviation estimate.\cite{Hughes2019-ud} However, both of these vertical methods can suffer from greatly reduced power compared to the parametric mixed effects models. Because SW-CRTs often have relatively few clusters, this can result in analyses that are highly underpowered to detect meaningful treatment effects.

In Section \ref{Methods}, we propose novel non-parametric methods to analyze SW-CRTs. In the first method, we, like Thompson et al., use within-period comparisons to avoid the problem of misspecification of time effects and cluster random intercept distributions. We incorporate the synthetic control procedure to match treated clusters with untreated clusters that are likely to be most similar. Synthetic controls are a relatively new but increasingly popular method for causal inference most common in the econometrics literature.\cite{Abadie2003-tp,Abadie2010-zy,Abadie2015-qi} The approach is generally used when there is one treated cluster and a ``donor pool'' of untreated clusters, with outcome data both before and after the treatment began. The method finds the linear combination of untreated clusters that most closely matches the pre-treatment outcomes of the treated cluster. The causal effect estimate is then some contrast of the treated cluster's post-treatment outcomes and that linear combination of the outcomes of the untreated clusters in the same period.\cite{Abadie2010-zy} We use this approach, somewhat akin to matching or covariate adjustment in parallel-arm clinical trials, to improve the power of the analysis.

The second method we propose uses the within-cluster between-period (``horizontal'') comparisons that are inherent in SW-CRT and other crossover designs to improve the power of non-parametric approaches.\cite{Hargreaves2015-zp,Thompson2017-or} This crossover method compares the between-period effect of clusters crossing over to that of clusters in the control arm in both periods (or in the intervention arm in both periods).

We also propose two ways of combining these methods. In one, we use synthetic controls to find the best-matching clusters for the crossover approach. In the other, we form an ensemble estimator by averaging the estimators obtained from the synthetic control and crossover methods.

In Section \ref{Results}, we compare by simulation the operating characteristics of these novel methods, the mixed effects models with both asymptotic and permutation-based inference, and the non-parametric within-period model, for both risk difference and odds ratio effect scales. We also apply these novel methods to a SW-CRT on the effects of diagnostic tests on tuberculosis outcomes reported by Trajman et al.,\cite{Trajman2015-mj} and compare the results to those for existing methods.

Finally, in Section \ref{Discussion}, we discuss the implications of these results for those designing and analyzing SW-CRTs. We also propose future research directions to better understand the relative performance of the methods considered here, as well as to better understand in which settings a SW-CRT may or may not be a reasonable design.

\section{Methods} \label{Methods}
In this section, we propose several novel methods of analysis for SW-CRTs: a synthetic control-based method, a crossover-based method, a combination method, and an ensemble method. These methods have flexible weighting schemes that allow the method to be tailored to particular situations. These methods do not rely on any particular distribution of the outcome data and can be used to estimate any causal contrast of interest.

\subsection{Setting and Notation}
Consider a SW-CRT with $I$ clusters with outcome measurements in each of $J$ periods. Denote by $Y_{i,j}$ the mean outcome for all $K$ measured individuals in cluster $i$ in period $j$ ($K$ can be fixed or vary by cluster-period). Let $X_{i,j}$ denote the intervention status of cluster $i$ in period $j$, with $X_{i,j} = 1$ indicating that the cluster is on intervention and $X_{i,j} = 0$ indicating that the cluster is on control. For each period $j$, let $I_{0,j} = \{i: X_{i,j} = X_{i,j-1} = 0\}$, $I_{1,j} = \{i: X_{i,j} = 1, X_{i,j-1} = 0\}$, and $I_{2,j} = \{i: X_{i,j} = X_{i,j-1} = 1\}$, the set of clusters on control in both periods $j$ and $j-1$, crossing over in period $j$, and on intervention in both periods $j$ and $j-1$, respectively. Denote the number of clusters in each of these sets by $n_{0,j}$, $n_{1,j}$, and $n_{2,j}$, respectively. We assume that each cluster only crosses over once; once a cluster is on intervention, it remains so for the rest of the periods under study. We assume throughout that the order of crossover is determined randomly. For each cluster $i$, let $j_i$ be the last period for which it is on control (define $j_i = 0$ if cluster $i$ is on intervention in period 1); then $j_i + 1$ is the first period on intervention for cluster $i$. For any period $j$, denote by $Y_{\cdot j}$ the expected value of the outcome (marginal across clusters) in period $j$ in the absence of intervention. That is, $Y_{\cdot j} = E[Y_{i,j}|X_{i,j} = 0]$ for any cluster $i$.

Let $g(y_1,y_2)$ be the contrast of interest. For example, for binary outcomes, the risk difference is given by $g(y_1,y_2) = y_1 - y_2$ and the log odds ratio is given by $g(y_1,y_2) = \log \left(\frac{y_1/(1-y_1)}{y_2/(1-y_2)} \right)$. Although binary outcomes are more common in SW-CRTs,\cite{Davey2015-no} contrasts on continuous and count outcomes may be specified as well.

\subsection{Existing methods for comparison} \label{Existing}
We compare the performance of these novel analysis methods with that of three current approaches for the analysis of SW-CRTs: two mixed effects model specifications (each with both asymptotic and exact inference) and the non-parametric within-period method.

First, we consider the commonly-used mixed effects model with a random intercept for cluster and fixed effects for time:
\begin{equation}
    h(E[Y_{i,j}]) = \mu + \alpha_i + \theta_j + X_{i,j} \beta, \label{MEM}
\end{equation}
where $h$ is the link function, $\mu$ is the global mean under control in period 1, $\alpha_i \overset{iid}{\sim} N(0,\tau^2)$, and $\theta_1 = 0$ for identifiability.\cite{Hussey2007-ay} Generalized linear mixed model theory can be used for asymptotic inference, and permutation tests (and associated confidence intervals) can be used for exact inference with this model.\cite{Wang2017-gr,Ji2017-cq}

Second, we consider the mixed effects model with an additional cluster-period random intercept:
\begin{equation}
    h(E[Y_{i,j}]) = \mu + \alpha_i + \theta_j + \eta_{i,j} + X_{i,j} \beta, \label{CPI}
\end{equation}
where $h$ is the link function, $\mu$ is the global mean under control in period 1, $\alpha_i \overset{iid}{\sim} N(0,\tau^2)$, $\theta_1 = 0$ for identifiability, $\eta_{i,j} \overset{iid}{\sim} N(0,\nu^2)$, and $\eta_{i,j} \indep \alpha_i$ for all $i,j$.\cite{Hooper2016-ih} Inference can proceed on an asymptotic or exact basis as above.\cite{Wang2017-gr}

Third, we consider the non-parametric within-period method. In this method, for each period $j$ where there are clusters on control and on intervention, a period-specific effect estimate is calculated by comparing the mean outcome of clusters on intervention ($i \in I_{1,j} \cup I_{2,j}$) to the mean outcome of clusters on control ($i \in I_{0,j}$):
\begin{equation}
    \hat{\beta}_j = g \left( \frac{\sum_{i \in I_{1,j} \cup I_{2,j}} Y_{i,j}}{n_{1,j} + n_{2,j}} , ~ \frac{\sum_{i \in I_{0,j}} Y_{i,j}}{n_{0,j}} \right). \label{NPWP1}
\end{equation}
The period-specific effect estimates are combined using an inverse-variance weighted average to obtain an overall estimated intervention effect:
\begin{equation}
    \hat{\beta} = \sum_{j: ~ 0 < n_{0,j} < I} \frac{w_j}{w} \hat{\beta}_j, \label{NPWP2}
\end{equation}
where $w_j = \left[ \left( \frac{(n_{0,j} - 1) s_{0,j}^2 + (n_{1,j} + n_{2,j} - 1) s_{1,j}^2}{J - 2} \right) \left( \frac{1}{n_{0,j}} + \frac{1}{n_{1,j} + n_{2,j}} \right)\right]^{-1}$, $w = \sum_{j: 0 < n_{0,j} < I} w_j$, and $s_{0,j}^2$ and $s_{1,j}^2$ are the empirical variances of the $Y_{i,j}$ values for clusters on control and on intervention, respectively, for period $j$.\cite{Thompson2018-so} A schematic representation of this estimation method is given in Figure \ref{Fig1}a. Exact inference can proceed using permutation tests.\cite{Thompson2018-so}

\subsection{Synthetic control method}
Our first proposed method uses the synthetic control procedure developed by Abadie et al.\ to estimate the effect of treatment for each intervention cluster-period.\cite{Abadie2010-zy} Similar to the non-parametric within-period method proposed by Thompson et al., this novel method constructs vertical comparisons and then uses a weighted average of these comparisons as an overall effect estimate.\cite{Thompson2018-so}

\begin{enumerate}
    \item For each period $j$ where there are clusters on control and on intervention, for each cluster $i$ on intervention ($i \in I_{1,j} \cup I_{2,j}$), we construct a synthetic control estimator $Z_{i,j}$, using the procedure outlined by Abadie et al.\cite{Abadie2010-zy} The synthetic control for cluster $i$ in period $j$ is a weighted average of the outcomes of the clusters on control in period $j$: $Z_{i,j} = \sum_{n=1}^{n_{0,j}} v_{i,j,n} Y_{m_n,j}$, where $m_1,\ldots,m_{n_{0,j}}$ are the clusters on control in period $j$. The weights, $v_{i,j,n}$, are selected by the synthetic control procedure to minimize the mean squared difference between the synthetic control for periods where cluster $i$ was on control and the true outcome for cluster $i$ in that period subject to the constraints that the weights are nonnegative and sum to one. That is, they minimize:
    \begin{equation} \label{MSPEdef}
        MSPE_{i,j} = \sum_{j':~X_{i,j'} = 0} \left(Y_{i,j'} - \sum_{n:~X_{m_n,j'} = 0} v_{i,j,n} Y_{m_n,j'} \right)^2.
    \end{equation}
See the proof of Theorem \ref{SCThm1} in Appendix \ref{app1} for details. If specific cluster-level covariates are known, they can be included in estimation of the synthetic control as well.\cite{Abadie2010-zy} When the synthetic control procedure does not converge or there are no pre-intervention periods for this cluster, the unweighted mean of the outcomes of clusters on control in period $j$ is used as $Z_{i,j}$. In these cases, the period-specific effect estimator is the same as that for the non-parametric within-period method described above, and so the properties of that estimator hold.
    \item For each intervention cluster $i$, for each period $j$ where $X_{i,j} = 1$ and $n_{0,j} \ge 1$, we construct an estimator:
    \begin{equation}
        \hat{\beta}_{i,j} = g(Y_{i,j},Z_{i,j}) \label{SC1}
    \end{equation}
    \item We find an overall estimator via a weighted average of these cluster-period-specific estimators:
    \begin{equation}
        \hat{\beta} = \sum_{j: ~ 0 < n_{0,j} < I}  ~ \sum_{i \in I_{1,j} \cup I_{2,j}} \frac{w_{i,j}}{w} \hat{\beta}_{i,j}, \label{SC2}
    \end{equation}
    where $w_{i,j} \ge 0$ and $w = \sum_{j: ~ 0 < n_{0,j} < I} ~ \sum_{i \in I_{1,j} \cup I_{2,j}} w_{i,j}$.
\end{enumerate}
 A schematic representation of this estimation method is given in Figure \ref{Fig1}b.

\subsubsection{Inferential procedure} \label{SCinference}
A permutation test can be used for exact inference, as for mixed effects models and the non-parametric within-period method.\cite{Wang2017-gr,Thompson2018-so} The standard permutation test approach is used: $P$ random permutations of the crossover order are generated and an estimate of the treatment effect is obtained from each permutation using the estimation procedure described above. The true estimate $\hat{\beta}$ is compared to these estimates and the $p$-value for the null hypothesis of no effect of treatment is given by the proportion of the $P$ estimates for which $\abs*{\hat{\beta}^p} \ge \abs*{\hat{\beta}}$. This approach matches inferential methods for synthetic control estimators, which rely on permuting the treatment indicator of units and estimating placebo synthetic control estimators to derive the null distribution of the estimator.\cite{Abadie2010-zy,Abadie2015-qi,Gautier2009-pn} To obtain confidence intervals, the permutation test can be inverted in the standard way.

\subsubsection{Computation}
This procedure is implemented in the \texttt{R} code available at:\\
\noindent\url{https://github.com/leekshaffer/SW-CRT-analysis}.\\
\noindent This implementation uses the \texttt{synth} function from the \texttt{Synth} package to obtain the synthetic control weights $v_{i,j,n}$.\cite{Abadie2011-ho}

\subsubsection{Properties of the estimator}
In a SW-CRT with a randomized order of crossover, the synthetic control estimator $Z_{i,j}$ is an unbiased estimate of the expected outcome under control, $Y_{\cdot j}$, if the underlying cluster-level outcome distribution is symmetric around some global mean outcome vector across periods; see Theorem \ref{SCThm1} in Appendix \ref{app1}. If the individual-level outcomes have cluster-conditional expectations that are symmetric around a global mean vector, the estimator is asymptotically unbiased as the number of subjects with measured outcomes per cluster increases; see Corollary \ref{SCCor1}. Thus, for any weights independent of the outcomes, the SC estimator using the risk difference is unbiased or asymptotically unbiased under these conditions if there is a common risk difference across cluster-periods. See Theorem \ref{SCThm2} and Corollary \ref{SCCor2}. Note that all of the assumptions of Corollary \ref{SCCor2} are satisfied under the mixed effects models described in Section \ref{Existing} with an identity link function as long as the random effects are independent and identically distributed following a normal (or any other symmetric) distribution.

\subsubsection{Selecting weights}
The weights for combining the cluster-period-specific estimators, $w_{i,j}$, can be selected on the basis of two primary goals: (1) minimizing the variance of the overall estimator $\hat{\beta}$; or (2) estimating a specific causal contrast when treatment effects may not be equal across clusters and time periods.

For the first goal, a natural approach is to follow the synthetic control literature on evaluating the accuracy of the synthetic control estimator or combining multiple synthetic control estimators by using the inverse of the mean squared prediction error (MSPE) values for each synthetic control estimator.\cite{Dube2015-pl,Donohue2018-qb,Powell2017-qf} For cluster $i$ in period $j$, the MSPE of the synthetic control fit is given by equation (\ref{MSPEdef}). In the SW-CRT setting, however, the MSPE values are not directly comparable as different synthetic control estimators have a different number of pre-intervention periods that contribute. By contrast, the MSPE values will be comparable for intervention clusters that begin treatment in the same period, as these clusters will always have the same number of pre-intervention periods, regardless of which of their intervention periods are being examined. We therefore propose to weight the $\hat{\beta}_{i,j}$ values by the inverse-MSPE within each set of intervention clusters that cross over in the same period, and then weight across these sets equally. That is, for each $(i,j)$ such that $X_{i,j} = 1$, set weights proportional to:
\begin{equation}
    w_{i,j} = \frac{MSPE_{i,j}^{-1}}{\sum_{(i',j'): i' \in I_{1,j}, X_{i',j'} = 1} MSPE_{i',j'}^{-1}}, \label{SCweights}
\end{equation}
where $MSPE_{i,j}$ is the MSPE of the synthetic control estimation procedure that produces $Z_{i,j}$.

For the second goal of weighting, the weighting approach will depend on the causal estimand of interest. If, for example, investigators are only interested in the effect of intervention in the first period of its introduction to any cluster, they may select as weights:
\begin{equation}
    w_{i,j} = \begin{cases} 1, &i \in I_{1,j} \\ 0, &\mbox{otherwise} \end{cases},
\end{equation}
that is, only using the $\hat{\beta}_{i,j}$ estimates for the first period on intervention for each cluster. We do not present results on this approach here, but further research is needed to understand the causal estimands that may be of interest when the treatment effect cannot be assumed to be constant across clusters and periods. In this way, the weights also aid interpretability of the estimator, as it is clear which clusters and periods are considered and how much weight is given to each.

\subsection{Crossover method}
The second novel method seeks to improve on the power of the non-parametric within-period method by incorporating horizontal comparisons at the time of crossover. There is substantial literature on the value of within-subject analysis methods and methods combining within- and between-subject analyses for individual randomized crossover trials, especially in the absence of anticipation, lag, or carryover effects of treatment.\cite{Everitt1995-cj,Jones1996-by,Omar1999-kf,Fitzmaurice2011-ab} The method we propose for SW-CRTs compares the mean contrast between the last control period and the first intervention period for each cluster crossing over from one period to the next to the mean contrast in those same periods among clusters on control in both periods. Since standard mixed effects models give a large weight to horizontal comparisons,\cite{Thompson2017-or} the crossover approach may recover some of the power of mixed effects models while preserving the robustness of non-parametric estimation. The procedure is as follows:
\begin{enumerate}
    \item For each cluster $i$ and period $j > 1$, define $D_{i,j} \equiv g(Y_{i,j},Y_{i,j-1})$, the contrast in outcomes in cluster $i$ between consecutive periods. E.g., for a risk difference analysis, $D_{i,j} = Y_{i,j} - Y_{i,j-1}$, the difference in outcomes between consecutive periods.
    \item For each period $j > 1$ with clusters on both intervention and control, estimate the treatment effect for period $j$ by:
    \begin{equation}
        \hat{\beta}_j = \sum_{i \in I_{1,j}} \frac{D_{i,j}}{n_{1,j}} - \sum_{i \in I_{0,j}} \frac{D_{i,j}}{n_{0,j}}. \label{CO1}
    \end{equation}
    If the treatment effect is assumed to be constant across time, an alternate estimator is given by:
    \begin{equation}
        \tilde{\beta}_j = \sum_{i \in I_{1,j}} \frac{D_{i,j}}{n_{1,j}} - \sum_{i \in I_{0,j} \cup I_{2,j}} \frac{D_{i,j}}{n_{0,j} + n_{2,j}}. \label{CO2}
    \end{equation}
    This alternative compares the change in outcome for the clusters which cross over to the change for clusters which either remain on control in both periods or remain on intervention in both periods.
    \item Construct an overall estimator with a weighted average of period-specific estimators:
    \begin{equation}
        \hat{\beta} = \sum_{j > 1:~0 < n_{0,j},n_{1,j} < I} \frac{w_j}{w} \hat{\beta}_j, \qquad \mbox{ or } \qquad \tilde{\beta} = \sum_{j > 1:~0 < n_{1,j} < I} \frac{w'_j}{w'} \tilde{\beta}_j, \label{CO3}
    \end{equation}
    where $w_j,w'_j \ge 0$ and $w = \sum_{j > 1:~0 < n_{0,j},n_{1,j} < I} w_j$ and $w' = \sum_{j > 1:~0 < n_{1,j} < I} w'_j$.
\end{enumerate}
 A schematic representation of this estimation method is given in Figure \ref{Fig1}c.

\subsubsection{Inferential procedure}
A permutation test can again be used for hypothesis testing and to obtain confidence intervals. The procedure is the same as the inferential procedure for the synthetic control estimator, detailed in Section \ref{SCinference}.

\subsubsection{Computation}
This procedure is implemented in the \texttt{R} code available at:\\
\noindent\url{https://github.com/leekshaffer/SW-CRT-analysis}.

\subsubsection{Properties of the estimator}
For the risk difference, $g(y_1,y_2) = y_1-y_2$, any of these crossover estimates are unbiased estimates of the true risk difference $\beta$, under a randomized crossover order and the assumption of a constant $\beta$ across clusters and periods. See Theorem \ref{COThm1} in Appendix \ref{app1}. The controls-only estimator $\hat{\beta}$ is unbiased for the intervention effect in the first period on intervention if that effect is constant across clusters. See Corollary \ref{COCor1} in Appendix \ref{app1}.

\subsubsection{Selecting weights} \label{COweights}
As for the synthetic control estimator, the weights can be selected either to minimize the variance of the overall estimator or to ensure proper estimation of a specific causal estimand. For the latter, again, this will depend on the specific estimand of interest, e.g., to match a target population of clusters.

To minimize the variance of the overall estimator, the weights may depend on the variance of the cluster-level outcome for each cluster-period. If all of these variances are assumed to be the same (i.e., all have the same subject-level variance and the same number of subjects), then the weight should depend only on the number of clusters in each treatment condition in that period. That is, we weight each estimator $\hat{\beta}_j$ by $w_j = \left(\frac{1}{n_{0,j}} + \frac{1}{n_{1,j}} \right)^{-1}$, the harmonic mean of the number of clusters used to estimate the consecutive-period control effect and the number of clusters used to estimate the crossover effect. For $\tilde{\beta}_j$, where the clusters which were on intervention in both periods $j$ and $j-1$ are used as control crossovers as well, we weight by $w'_j = \left(\frac{1}{n_{0,j}+n_{2,j}} + \frac{1}{n_{1,j}} \right)^{-1}$. Note that when the same number of clusters cross over at each time point, $w'_j$ is constant across $j$, while $w_j$ decreases as $j$ increases.

\subsection{Crossover-synthetic control method}
A third potential method combines these two approaches by finding a synthetic control for the horizontal crossover contrast and comparing the intervention horizontal contrast to this synthetic control. This may combine the benefits of using horizontal comparisons with the benefits of synthetic control-based matching between clusters.

\begin{enumerate}
    \item For each cluster $i$ and period $j > 1$, define $D_{i,j} \equiv g(Y_{i,j},Y_{i,j-1})$, the contrast in outcomes in cluster $i$ between consecutive periods. E.g., for a risk difference analysis, $D_{i,j} = Y_{i,j} - Y_{i,j-1}$, the difference in outcomes between consecutive periods.
     \item For each cluster $i \in I^*  \equiv \{ i: j_i + 1 > 1 \cap n_{0,j_i + 1} > 0 \}$, the set of clusters that begin intervention in a period after period 1 that has clusters on control, construct a synthetic control horizontal contrast estimator $C_{i}$, using the procedure outlined by Abadie et al.\cite{Abadie2010-zy} For cluster $i$, which crosses over in period $j$, the synthetic control is a weighted average of the horizontal contrasts of the clusters on control in both periods $j-1$ and $j$: $C_{i} = \sum_{n=1}^{n_{0,j}} v_{i,n} D_{m_n,j}$, where $m_1,\ldots,m_{n_{0,j}}$ are the clusters on control in both periods. The weights are selected by the synthetic control procedure to minimize the mean squared difference between the synthetic control for periods $j'$ where cluster $i \in I_{0,j'}$ and the true horizontal contrast for cluster $i$ in that period subject to the constraints that the weights are nonnegative and sum to one. When the synthetic control procedure does not converge or there are no pre-crossover consecutive period contrasts for this cluster, the unweighted mean of the values $D_{i',j}$ for $i' \in I_{0,j}$ is used as $C_i$.
    \item For each cluster $i \in I^*$, we construct an estimator using its crossover effect in period $j$:
    \begin{equation}
    \hat{\beta}_i = D_{i,j} - C_{i}. \label{COSC1}
    \end{equation}
    \item We find an overall estimator via a weighted average of these cluster-specific estimators:
    \begin{equation}
    \hat{\beta} = \sum_{i \in I^*} \frac{w_i}{w} \hat{\beta}_i, \label{COSC2}
    \end{equation}
    where $w = \sum_{i \in I^*} w_i$.
\end{enumerate}
 A schematic representation of this estimation method is given in Figure \ref{Fig1}d. Note that this procedure is the same as that for the synthetic control method, but using $D_{i,j}$ as the ``outcomes'' in place of $Y_{i,j}$.

\subsubsection{Inferential procedure}
The inferential procedure for the synthetic control estimator, detailed in Section \ref{SCinference}, can again be used here for exact inference.

\subsubsection{Computation}
This procedure is implemented in the \texttt{R} code available at:\\
\noindent\url{https://github.com/leekshaffer/SW-CRT-analysis}.\\
\noindent This implementation uses the \texttt{synth} function from the \texttt{Synth} package to obtain the synthetic control weights $v_{i,n}$.\cite{Abadie2011-ho}

\subsubsection{Selecting weights}
As for the synthetic control estimator, a natural approach to minimize the variance of the overall estimator is to use weights inversely proportional to the MSPE of the synthetic control fits. Again, though, because the number of pre-crossover periods varies, these are only comparable among clusters which cross over in the same period. So we propose to weight the $\hat{\beta}_i$ values by the inverse-MSPE within each set of intervention clusters that cross over in the same period, and then weight across these sets equally. That is, for each $i$, set:
\begin{equation}
    w_{i} = \frac{MSPE_{i}^{-1}}{\sum_{i' \in I_{1,j_i + 1}} MSPE_{i'}^{-1}}, \label{COSCweights}
\end{equation}
where $MSPE_{i}$ is the MSPE of the synthetic control estimation procedure that produces $Z_{i}$.

\subsection{Ensemble method}
Finally, we consider an ensemble method that averages the estimators of previously-described methods. For any unbiased and/or consistent estimators, a weighted average of those estimators with weights that do not depend on the data will also be unbiased/consistent. If the covariance of the estimators is small enough compared to the variances, it may also reduce the variance of the estimator. In Appendix \ref{app2}, we derive the variances and covariance of the non-parametric within-period and crossover estimators under a simplified data-generating process. We then demonstrate that in a simplified data-generating setting, when the difference in the mean outcome between clusters is relatively small, compared to the variability within clusters, a simple mean of the non-parametric within-period estimator and the crossover estimator has a lower variance than either estimator on its own.

There is no known analytic formula for the variance of the synthetic control estimator, although we expect (and simulation results presented below suggest) the synthetic control estimator to have lower variance than the non-parametric within-period estimator when the synthetic control matching performs well. Since the synthetic control and non-parametric within-period estimators are both vertical methods of analysis, we consider here an ensemble estimator that is a simple mean of the synthetic control estimator and the crossover estimator. That is,
\begin{equation}
    \hat{\beta}^{ENS} = \frac{1}{2} \hat{\beta}^{SC} + \frac{1}{2} \hat{\beta}^{CO}, \label{ENS}
\end{equation}
where $\hat{\beta}^{SC}$ is a synthetic control estimator and $\hat{\beta}^{CO}$ is a crossover estimator.

Note that many other ensemble estimators could be constructed, using different analysis methods and different weights. In addition, within-period ensembles may be constructed and then combined across periods (e.g., take the average of the SC estimators within each period $j$ and average those with the CO estimator for period $j$, and then combine across periods to target a specific causal estimand). We use this simple version here to demonstrate the concept of the ensemble method and show its potential to improve power, but different ensembles will have different operating characteristics and may perform better or worse depending on the setting.

\subsubsection{Inferential procedure}
The inferential procedure for the synthetic control estimator, detailed in Section \ref{SCinference}, can again be used here for exact inference.

\subsubsection{Computation}
This procedure is implemented in the \texttt{R} code available at:\\
\noindent\url{https://github.com/leekshaffer/SW-CRT-analysis}.\\
\noindent Other ensemble methods can be constructed by altering the weights and estimators used; a generic function is provided for this purpose in the \texttt{R} code.

\section{Results} \label{Results}

We compare the performance of these novel methods with the existing methods under two simulation settings: the first using the risk difference contrast, $g(y_1,y_2) = y_1 - y_2$, and the second using the log odds ratio contrast, $g(y_1,y_2) = \log \left( \frac{y_1/(1-y_1)}{y_2/(1-y_2)} \right)$. As SW-CRTs most commonly have binary outcomes, we consider binary outcomes here; the methods, however, also work for continuous outcomes. Throughout we denote the methods considered as follows:
\begin{itemize}
    \item MEM denotes the mixed-effects model defined in equation (\ref{MEM}).
    \item CPI denotes the mixed-effects model with a cluster-period random effect defined in equation (\ref{CPI}).
    \item NPWP denotes the non-parametric within-period method defined in equations (\ref{NPWP1}) and (\ref{NPWP2}).
    \item SC-1 denotes the synthetic control method defined in equations (\ref{SC1}) and (\ref{SC2}), with equal weights across cluster-period estimators.
    \item SC-2 denotes the synthetic control method with inverse-MSPE weights as defined in equation (\ref{SCweights}). In this case, there is only one cluster crossing over per period, so the estimators are weighted by inverse-MSPE within each target cluster, and then equally weighted across clusters.
    \item CO-1 denotes the crossover method defined in equations (\ref{CO1}) and (\ref{CO3}), using comparison data only from control clusters, with equal weights.
    \item CO-2 denotes  the crossover method defined in equations (\ref{CO1}) and (\ref{CO3}), using comparison data only from control clusters, with weights proportional to the harmonic mean of the number of control and crossover clusters.
    \item CO-3 denotes the crossover method defined in equations (\ref{CO2}) and (\ref{CO3}), using comparison data from both control clusters and intervention clusters, with equal weights.
    \item COSC-1 denotes the crossover-synthetic control method defined in equations (\ref{COSC1}) and (\ref{COSC2}) with equal weights across cluster-specific estimators.
    \item COSC-2 denotes the crossover-synthetic control method with inverse-MSPE weights defined in equation (\ref{COSCweights}).
    \item ENS denotes the ensemble method defined in equation (\ref{ENS}), using a simple mean of SC-2 and CO-2.
\end{itemize}
All inference is based on exact permutation tests, except for asymptotic inference using the MEM and CPI models, which is denoted by MEM-a and CPI-a. All permutation tests were conducted with $500$ randomly-sampled permutations of the crossover order.

\subsection{Simulation 1: risk difference}

\subsubsection{Setting and parameters}
We consider a setting where the risk difference is the contrast of interest. There are $I=7$ clusters and $J=8$ time periods, with one cluster beginning treatment in each of periods $2$ through $8$. At each cluster-period, $K=100$ individuals are sampled. The data are generated from a mixed effects model similar to that in equation (\ref{CPI}) with $\mu = 0.30$ and $\tau = 0.06$, with an identity link. We consider four scenarios:
\begin{enumerate}
    \item Fixed time effects $\bm{\theta} = \bm{\theta}_1 \equiv (0,0.08,0.18,0.29,0.30,0.27,0.20,0.13)$ and no cluster-period effect ($\nu = 0$). The MEM model is correctly specified in this case.
    \item Fixed time effects $\bm{\theta} = \bm{\theta}_1$ and cluster-period effect with $\nu = 0.01$. The CPI model is correctly specified in this case.
    \item Equal probability of each cluster having either the time effects $\bm{\theta}_1$ or \\ $\bm{\theta}_2 \equiv (0,0.02,0.03,0.07,0.13,0.19,0.27,0.3)$. No cluster-period effect ($\nu = 0$). Neither MEM nor CPI is correctly specified in this case.
    \item Equal probability of each cluster having either the time effects $\bm{\theta}_1$ or $\bm{\theta}_2$. Cluster-period effect with $\nu = 0.01$. Neither MEM nor CPI is correctly specified in this case.
\end{enumerate}
Note that all scenarios satisfy the conditions of Corollaries \ref{SCCor1} and \ref{SCCor2} and Theorem \ref{COThm1} in Appendix \ref{app1}, so the SC-1 estimator is asymptotically unbiased and the CO estimators are unbiased. Since SC-2 does not have equal weights, it does not meet the conditions of Theorem \ref{SCThm2} or Corollary \ref{SCCor2}, so we cannot guarantee it is asymptotically unbiased. For scenarios 1 and 2, the global mean vector is $\bm{Y}_{\cdot J} = \mu + \bm{\theta}_1$. For scenarios 3 and 4, the global mean vector is $\bm{Y}_{\cdot J} = \mu + \frac{\bm{\theta}_1 + \bm{\theta}_2}{2}$.

For each scenario, 1,000 data sets were simulated for each of three treatment effects: $\beta = -0.2$, $\beta = -0.1$, and $\beta = 0$. We do not present the results for the strong treatment effect ($\beta = -0.2$) here, as they are very similar to those for the moderate treatment effect ($\beta = -0.1$), but with such high power (all methods except NPWP over 90\% in all scenarios) that it is hard to distinguish differences. A representative plot of cluster outcomes for each of the four scenarios with no treatment effect is given in Figure 2. If the probability of outcome for any cluster-period was less than 0 or greater than 1, it was truncated to 0 or 1, respectively. Code to generate and analyze the simulated data is available at:\\
\noindent\url{https://github.com/leekshaffer/SW-CRT-analysis}.

For each of the twelve scenarios, each data set was analyzed using the following methods: MEM, CPI, NPWP, SC-1, SC-2, CO-1, CO-2, CO-3, COSC, and ENS. Note that since only one cluster crosses over in each period, COSC-2 is equivalent to COSC-1; this is denoted COSC. The weights for SC-2 are calculated with inverse-MSPE weighting only within each intervention cluster but still differ from SC-1, which is equally weighted among all cluster-periods.

\subsubsection{Simulation results}
Figure \ref{Fig3} shows the mean effect estimate and 1/2-standard deviation of the effect estimates across the 1,000 simulations for each method for each scenario. The two subplots each show the scenarios for one treatment effect, with Scenario 1 at the top and Scenario 4 at the bottom of each plot. For all of the settings, all of the methods exhibit little overall bias, with the average estimate for each method within $0.005$ of the true effect in each scenario. As expected given that all four scenarios meet the assumptions of Corollary \ref{SCCor2} and Theorem \ref{COThm1}, SC-1, CO-1, and CO-2 appear to be unbiased in the simulations. As noted by Thompson et al.,\cite{Thompson2018-so} the nonparametric estimator NPWP must also be unbiased in all scenarios. Despite the misspecification of MEM and CPI in scenarios 3 and 4, they nonetheless result in unbiased estimators, albeit with wider empirical variance. And SC-2 appears unbiased in these simulations as well, despite its not meeting the conditions of Corollary \ref{SCCor2}. The variability of the effect estimates varies a great deal by method, with the MEM and CPI methods exhibiting the least variability when the time effects do not vary, and the CO and ENS methods exhibiting the least variability when the time effects do vary. Figure \ref{Fig4} shows the Type I Error (probability of finding a significant treatment effect when $\beta = 0$) for each analysis method under each scenario. All of the methods are close to the nominal Type I Error of 5\% with the exception of asymptotic inference for the MEM and CPI methods when the time effects vary. All of the exact inference methods also achieve the nominal coverage in 95\% confidence intervals, as shown in Figure \ref{Fig5}. When the time effects do not vary, the MEM and CPI methods with asymptotic inference also achieve or nearly achieve the nominal coverage; when the time effects do vary, they both have less than 90\% coverage.

Figure \ref{Fig6} shows the power (estimated probability of finding a significant treatment effect at the 5\% significance level) for each analysis method under each scenario for the moderate treatment effect ($\beta = -0.1$). The asymptotic inference MEM and CPI results are not shown when the time effects vary as they have inflated Type I Error. The MEM and CPI (exact or asymptotic inference) methods have the highest power when the time effects do not vary. When the time effects do vary, the CO and ENS methods perform the best among the exact inference methods, followed by the COSC, SC, MEM, and CPI methods. The NPWP method has the least power. As expected with weights selected to reduce variance, CO-2 outperforms CO-1 and SC-2 outperforms SC-1. These differences, however, are smaller than the differences between classes of methods.

\subsection{Simulation 2: odds ratio}

\subsubsection{Setting and parameters}
We consider now a setting where the odds ratio is the contrast of interest. There are again $I=7$ clusters and $J=8$ time periods, with one cluster beginning treatment in each of periods $2$ through $8$. At each cluster-period, $K=100$ individuals are sampled. The data are generated from a mixed effects model similar to that in equation (\ref{CPI}) with $\mu = \mbox{logit}(0.30)$ and $\tau = 0.1$, with a logit link. We consider four scenarios:
\begin{enumerate}
    \item Fixed time effects $\bm{\theta} = \bm{\theta}_1 \equiv \log(1,1.43,2.15,3.36,3.50,3.09,2.33,1.76)$ and no cluster-period effect ($\nu = 0$). The MEM model is correctly specified in this case.
    \item Fixed time effects $\bm{\theta} = \bm{\theta}_1$ and cluster-period effect with $\nu = 0.01$. The CPI model is correctly specified in this case.
    \item Equal probability of each cluster having either the time effects $\bm{\theta}_1$ or \\ $\bm{\theta}_2 \equiv \log(1,1.10,1.15,1.37,1.76,2.24,3.09,3.50)$. No cluster-period effect ($\nu = 0$). Neither MEM nor CPI is correctly specified in this case.
    \item Equal probability of each cluster having either the time effects $\bm{\theta}_1$ or $\bm{\theta}_2$. Cluster-period effect with $\nu = 0.01$. Neither MEM nor CPI is correctly specified in this case.
\end{enumerate}
For each scenario, 1,000 data sets were simulated for each of three treatment effects: $\beta = \log(0.50) \approx -0.693$, $\beta = \log(0.66) \approx -0.416$, and $\beta = \log(1) = 0$. We do not present the results for the strong treatment effect ($\beta = \log(0.50) \approx -0.693$) here, as they are very similar to those for the moderate treatment effect ($\beta = \log(0.66) \approx -0.416$), but with such high power (all methods over 90\% in all scenarios) as to make comparisons difficult. These parameters were chosen to give similar outcome probabilities under control as in Simulation 1, but specified on the log-odds ratio scale. A representative plot of cluster outcomes for each of the four scenarios with no treatment effect is given in Figure \ref{Fig7}. Code to generate and analyze the simulated data is available at:\\
\noindent\url{https://github.com/leekshaffer/SW-CRT-analysis}.

For each of the twelve scenarios, each data set was analyzed using the same set of methods as in the previous section.

\subsubsection{Simulation results}
Figure \ref{Fig8} shows the mean effect estimate and 1/2-standard deviation of the effect estimates across the 1,000 simulations for each method for each scenario. The three subplots each show the scenarios for one treatment effect, with Scenario 1 at the top and Scenario 4 at the bottom of each plot. For all of the settings, all of the methods exhibit little overall bias, with the average estimate for each method within $0.01$ of the true effect in each scenario. The variability of the effect estimates varies a great deal by method, however, with MEM and CPI methods exhibiting the least variability when the time effects do not vary, and the CO, SC, and ENS methods exhibiting the least variability when the time effects do vary. Figure \ref{Fig9} shows the Type I Error (probability of finding a significant treatment effect when $\beta = 0$) for each analysis method under each scenario. All of the methods are close to the nominal Type I Error of 5\% with the exception of asymptotic inference for the MEM and CPI methods when the time effects vary. All of the exact inference methods also achieve the nominal coverage in 95\% confidence intervals, as shown in Figure \ref{Fig10}. When the time effects do not vary, the MEM and CPI methods with asymptotic inference also achieve or nearly achieve the nominal coverage; when the time effects do vary, they both have less than 90\% coverage.

Figure \ref{Fig11} shows the power for each analysis method under each scenario for the moderate treatment effect ($\beta = \log(0.66) \approx -0.416$). The asymptotic inference MEM and CPI results are not shown when the time effects vary as they have inflated Type I Error. The MEM and CPI (exact or asymptotic inference) methods have the highest power when the time effects do not vary, but there is relatively little loss of power for the ENS, SC, and CO-3 methods. When the time effects do vary, the ENS method performs the best among the exact inference methods, followed by the CO-3 method, the SC and other CO methods, and then the COSC method. The NPWP and exact inference MEM and CPI methods have the least power.

These results are largely similar to those seen in Simulation 1. This suggests that the contrast of interest is less important to the relative performance of these methods than the underlying distribution of the data.

\subsection{Variance and covariance of estimators}
To assess the variability between methods for a given instance of analysis, we determined the pairwise covariance for each pair of methods across the simulated settings. Within each data-generating setting, we found the covariance between methods across all 1,000 simulations. As a representative example of these covariances, we take scenario 4, the scenario with the most complex data-generating process, under the null hypothesis of no treatment effect, for both Simulation 1 and Simulation 2. The covariances are displayed in the heat map shown in Figure \ref{Fig12} for the risk difference (Simulation 1) and in Figure \ref{Fig13} for the odds ratio (Simulation 2). Note that the empirical variances for each method across the 1,000 simulations are on the diagonal of the figure.

These results indicate rather high correlations within classes of methods; that is, the mixed effects model methods are highly correlated with one another, the synthetic control methods are highly correlated with one another, and the crossover methods (including COSC) are highly correlated with one another. NPWP is correlated with the mixed effects model methods and the SC methods; it has high covariance with these methods largely because of its high variance. 

The least covariance occurs between any mixed effects or vertical (NPWP or SC) method and any of the CO-based methods. This suggests that using an ensemble method combining an SC method and a CO method is indeed valuable here, as, for example, the covariance of SC-2 and CO-2 is notably smaller than the variance of SC-2 and the variance of CO-2. This corresponds with the increased power for the ENS method compared with SC-2 and CO-2 seen in the previous sections. The relatively low, equal covariance of ENS with the other methods suggests there is little to be gained in this setting by more complex ensemble methods.

\subsection{Application to tuberculosis SW-CRT} \label{Example}
We applied the methods discussed here to a SW-CRT that assessed the effect of a tuberculosis (TB) diagnostic test on reducing unsuccessful (non-cure) outcomes of adults on TB treatment.\cite{Trajman2015-mj} Note that this is the same trial re-analyzed by Thompson et al.\ using the within-period methods they proposed.\cite{Thompson2018-so}

\subsubsection{Trial description}
In this study, Trajman et al.\ conducted a SW-CRT in fourteen laboratories in the Brazilian cities of Rio de Janeiro and Manaus. While in the control arm, the labs diagnosed TB using two-sample sputum smear microscopy; in the intervention arm, diagnosis and first-line evaluation of potential drug resistance was by a single sputum sample XpertMTB/RIF assay. Data were collected on individuals diagnosed with TB in eight months in 2012 in the clinics associated with these laboratories. In the first month, all labs were in the control arm. In each subsequent month, two labs were switched to the intervention arm. In the final month, all labs were in the intervention arm.\cite{Durovni2014-nq}

The outcome of interest was the proportion of unfavorable TB treatment outcomes, where unfavorable outcomes are defined as: loss to follow-up, TB-attributed death, death from other causes, change of diagnosis, transfer out (including to specialized clinics for management of drug-resistant TB or drug intolerance), and suspicion of drug resistance. In total, the trial analyzed the intervention and outcome status of 4,054 patients.\cite{Trajman2015-mj}

\subsubsection{Goodness of fit of mixed effects models}
Before analyzing these data using non-parametric approaches, we consider the goodness of fit of the mixed effects models. We fit both the MEM and CPI models, as usual assuming independent normally-distributed random effects. In this case, the CPI model yields nearly the same fitted values as the MEM model, so we consider only the MEM model from this point. A variety of methods have been proposed to assess the assumption of independent normally-distributed random effects.\cite{Alonso2008-nm,Huang2009-vi,Verbeke2013-oq,Meintanis2016-al,Drikvandi2017-rf,Singer2013-op,Ritz2004-yq} We use several of these methods to assess the assumption in this case; details are in Appendix \ref{app3}. Some methods indicate a violation of the assumption and others do not, but caution should be exercised in interpreting these results as diagnostic tests may not be powerful or reliable for such a small number of clusters.\cite{Yap2011-bu} Because of the potential of model misspecification, we proceed with the non-parametric analyses.

\subsubsection{Results}
The primary analysis conducted by Trajman et al., which did not adjust for time effects, found a decrease in the number of events (unsuccessful outcomes) in the intervention arm compared to the control arm, although this decrease was not statistically significant at the 0.05 level.\cite{Trajman2015-mj} Re-analyzing the data using the NPWP method, Thompson et al.\ found a statistically significant decrease on both the odds ratio and risk difference scales.\cite{Thompson2018-so}

We analyzed these data using all of the methods described here using both the risk difference and log odds ratio contrasts; all exact inference methods use 500 permutations. Note that the NPWP method corresponds to that used by Thompson et al.\ for the risk difference scale. For the odds ratio, we use the log odds ratio contrast and exponentiate after averaging across periods for comparability with the SC and CO methods, whereas Thompson et al.\ calculate on the odds ratio scale directly, leading to a slight difference in the estimate.\cite{Thompson2018-so} In both cases, inference may differ slightly because of the stochasticity in the permutation-based inference. This stochasticity, as well as the difference of calculating under the alternative hypothesis rather than the null, can also lead to confidence intervals including the null when the hypothesis test rejects the null and vice versa. Also note that since there are two clusters which cross over at each time period, COSC-1 and COSC-2 yield different results. The results, reported on the risk difference scale and the odds ratio scale, are summarized in Table \ref{Tbl:TB}.

\begin{table}[!ht]
\caption{Results from SW-CRT of diagnostic method on rates of unfavorable TB treatment outcomes in Brazil, by analysis method}  \label{Tbl:TB}
\centering
\begin{tabular}{|l|llc|llc|} 
\toprule
    & \multicolumn{3}{c|}{Risk Difference} & \multicolumn{3}{c|}{Odds Ratio} \\
   Method & Estimate & 95\% Conf. Int. & $p$ & Estimate & 95\% Conf. Int. & $p$ \\
   \midrule
   MEM/CPI & -3.59\% & (-8.9\%, 1.4\%) & 0.126 & 0.835 & (0.66, 1.07) & 0.104 \\
   MEM/CPI-a & -3.59\% & (-8.4\%, 1.1\%) & 0.105 & 0.835 & (0.66, 1.05) & 0.091 \\
   NPWP & -4.83\% & (-10.1\%, 0.1\%) & 0.050 & 0.794 & (0.61, 0.99) & 0.046 \\
   SC-1 & -7.28\% & (-18.2\%, 1.0\%) & 0.084 & 0.703 & (0.44, 1.04) & 0.066 \\
   SC-2 & -8.29\% & (-18.3\%, 1.1\%) & 0.080 & 0.675 & (0.43, 1.07) & 0.082 \\
   CO-1 & -7.34\% & (-14.5\%, 0.5\%) & 0.064 & 0.703 & (0.49, 1.04) & 0.046 \\
   CO-2 & -6.97\% & (-14.0\%, 0.5\%) & 0.052 & 0.717 & (0.50, 1.03) & 0.054 \\
   CO-3 & -7.00\% & (-14.0\%, 0.0\%) & 0.050 & 0.721 & (0.51, 1.00) & 0.036 \\
   COSC-1 & -7.01\% & (-15.5\%, 1.1\%) & 0.078 & 0.728 & (0.49, 1.10) & 0.118 \\
   COSC-2 & -5.12\% & (-14.7\%, 4.5\%) & 0.242 & 0.784 & (0.50, 1.18) & 0.222 \\
   ENS & -7.63\% & (-15.0\%, -0.6\%) & 0.036 & 0.696 & (0.49, 0.95) & 0.032 \\
\bottomrule
\end{tabular}
\end{table}

The novel methods identify a stronger treatment effect than do the model-based and NPWP methods. As Thompson et al. show, the NPWP method here places a large amount of weight on the contrast in the fifth period, which has a modest (-2.23\%) effect.\cite{Thompson2017-or} This attenuates the effect compared to, for example, CO-1, which equally weights contrasts in different periods. It also, however, reduces the variance of the overall estimator, thus yielding a lower $p$-value for the NPWP method than the CO methods which use the control crossovers only. The COSC methods do not appear to give more precision than the CO methods, but yield similar effect estimates. On both scales, the ENS method yields the lowest $p$-values, as it detects a strong effect and has more precision than the other novel methods. All of the results suggest a protective effect of the intervention, with the novel methods detecting a larger effect but with more uncertainty, and the NPWP method estimating a narrower confidence interval of smaller effect sizes. In this example, inference does indeed depend upon the analysis method used, demonstrating the need for consideration of the assumptions underlying each method.

\section{Discussion} \label{Discussion}

These results demonstrate the potential of analytic methods for SW-CRTs that do not rely on parametric modeling of secular trends for validity. These methods achieve greater power than the purely vertical within-period method by using the history of outcomes within each cluster inherently collected in a SW-CRT to match the most similar clusters or by using horizontal, within-cluster information. In the simulation settings used here, when the mixed effects models were misspecified, an ensemble method that averaged the crossover method and the synthetic control method had the highest power to detect a true treatment effect, followed by the crossover method. Further research is needed to determine in which settings each of these methods is likely to perform the best, and whether an ensemble method can be constructed to be powerful across settings. These results demonstrate that this simple ensemble method may in some settings perform better, but not that it always will or that this is the most powerful ensemble method for any particular setting. The potential for incorporating measured covariates or stratified randomization into the SC method may also lead to increased power in some situations.

While these methods are valuable and in general rely on weaker assumptions than mixed effects methods do for unbiasedness, they are still not as powerful as parametric mixed effects methods when the modeling assumptions are met. This leaves an important role for investigators to determine when assumptions are likely to be met and for research on secular trends in the conditions and interventions commonly studied by SW-CRTs to determine when non-parametric methods are needed. Additionally, further work on using regression diagnostics to identify violations of modeling assumptions would be very valuable. Investigators should consider exact inference on parametric methods when the modeling assumptions of mixed effects methods are likely to nearly hold and the non-parametric methods when the secular trends are unknown or the modeling assumptions methods are likely to be strongly violated. Caution should be exercised regarding the SC methods as well when the underlying data distribution is unlikely to be symmetric or asymptotically symmetric.

As can be seen in the imperfect correlation across methods and the variability of the estimates, with relatively few clusters, the estimation can be very sensitive to the analysis method, and even the weighting scheme, chosen. The performance of any method in one particular analysis of a trial may not reflect its overall operating characteristics. The specific settings where the estimators depend heavily on certain cluster-periods and the impact that has on operating characteristics deserve more scrutiny. Again, this is an area where ensemble methods may prove useful in mitigating high dependence on specific cluster-periods by certain methods.

The methods presented here also provide advantages in interpretability and flexibility. When the treatment effect is not constant across clusters or across time periods, the mixed effects model estimate for nonlinear link functions is a conditional parameter, and its interpretation can be unclear.\cite{Hubbard2010-ly} For linear link functions, the mixed effects model estimate is a weighted average intervention effect that depends on the form of the time trend in the treatment effect.\cite{Nickless2018-cm} With the non-parametric methods, using equal weighting across clusters and periods, the estimate is easily interpreted as an average treatment effect across cluster-periods in the study. Other causal effects can be estimated using weights chosen to match the target parameter, depending on the effect of interest and assumptions the investigators are willing to make about generalizability to a separate target population. More work is needed to determine how to select weights that maximize efficiency for specific causal parameters that may be of common interest. For instance, if the effect of time on the intervention effect is known or a parametric form can be assumed, there may be an efficiency-maximizing weighting scheme.

When treatment effects are not instantaneous---common in settings where treatment effects vary over time---methods must be modified. Throughout this article, we have assumed that the full effect of treatment occurs during the first period of treatment and that there are no anticipation effects prior to that point. In practice, it may be desirable to account for a lag in, or gradual onset of, treatment effects resulting from logistical complexity in reaching everyone in the cluster or for the effect to reach its full strength.\cite{Hemming2015-yq,Copas2015-so} This can be incorporated into the SC methods by taking as the time of start of the intervention the time of completion of such a transition period. It can be incorporated into the CO methods by taking as the ``crossover effect'' the contrast between the first period after the transition and the last period prior to any anticipation effects. Achieving the same efficiency as would be achieved with a similar trial with no transition period may require more clusters or more time between successive cluster crossovers. All SW-CRT methods are sensitive to properly accounting for the transition period, but the CO methods are particularly sensitive because of their focus on the horizontal comparison. If the transition period length is unknown or likely to vary across clusters, the CO methods may not be appropriate.

The synthetic control method allows for additional flexibility and the potential for increased power and use in a wider variety of settings. As mentioned above, it can be useful when lagged treatment effects or time-varying treatment effects make a specific causal estimand more desirable as a target for inference. It also, as shown in the simulations here, can be a valuable part of an ensemble method that improves the power of an estimator. And for trials with more periods, or a longer pre-intervention history, the SC method itself may perform better. In general, it provides many of the advantages of the non-parametric within-period method while using a matching-like procedure to increase power. For the COSC method, the relatively poor performance in these simulation settings may stem from the fact that one period of history is lost by using the crossover estimator. With few periods, that can have a large effect on the power. Again, a longer pre-intervention history may improve the value of this method.

Additionally, more advanced techniques can be used to improve synthetic control matching and thus potentially improve the power of the SC and COSC methods. Synthetic controls can incorporate measured covariates to improve the matching.\cite{Abadie2010-zy,Botosaru2019-rg} Moreover, new synthetic control algorithms and methodologies may also be useful in improving the matching and designing efficiency-maximizing weighting schemes. These include Bayesian synthetic control approaches,\cite{Bruhn2017-on,Kim2019-gp} flexible non-parametric synthetic control,\cite{Cerulli2019-mh} generalized synthetic control,\cite{Xu2017-ne} and augmented synthetic control.\cite{Ben-Michael2018-bq} The SC method, potentially incorporating these approaches to improve the causal inference component, may also provide a path for analysis of non-randomized studies that mimic stepped wedge trials, as the synthetic control may address confounding of treatment initiation. Further work is needed in this area to determine whether the stepped wedge trial design can be used as a target trial for causal inference from observational studies.\cite{Hernan2016-aj,Garcia-Albeniz2017-wk}

These methods increase the number of analysis options available to investigators conducting stepped wedge cluster randomized trials. The SC method provides a semi-parametric option that relies on weaker assumptions on the underlying data-generating process than mixed effects models, while increasing power compared to the NPWP method, and it can be improved with advanced methods or with additional pre-intervention data. The CO method provides a non-parametric option with greatly improved power, although it relies on a constant treatment effect that appears very soon after treatment initiation. Variations of these methods and ensemble methods can also be used to target specific causal parameters and improve power in certain circumstances. Careful consideration is still required, however, to determine which analysis method is most appropriate for each individual circumstance, and more work is needed to clarify how to make that determination \emph{a priori} or in a systematic way. Moreover, careful selection of analysis method does not alleviate all of the drawbacks and concerns about SW-CRTs and, as mentioned above, does not ensure ideal performance of any single analysis. Investigators should continue to select the appropriate trial design for each study, taking into account analysis methods, the target estimand, and power considerations, along with issues of logistical feasibility, ethics, risk-benefit profiles, and generalizability.

\subsection*{Financial disclosure}

Research reported in this publication was supported by the National Institute of Allergy and Infectious Diseases under Award Numbers 5T32AI007358-28 and 1F31AI147745 (for L.K.S.), and R3751164 (for V.d.G.); and National Institute of General Medical Sciences Award Number U54GM088558 (for M.L.).

\subsection*{Conflict of interest}

The authors declare no potential conflicts of interest.

\subsection*{Code Availability}

\texttt{R} code to implement the methods detailed in this article is available at:\\
\noindent\url{https://github.com/leekshaffer/SW-CRT-analysis}.\\
\noindent The code is still in progress as we work to improve usability and speed of implementation. Additionally, code to replicate the simulations detailed here and to generate the figures from those simulations is available at the same source.

\bibliography{main}

\begin{thebibliography}{10}
\providecommand \doibase [0]{http://dx.doi.org/}%

\bibitem{Halloran2010-zd}
Halloran ME, Longini IM, Struchiner CJ. {\it Design and Analysis of Vaccine
  Studies}.
\newblock Statistics for Biology and Health. New York: Springer;
\newblock 2010.

\bibitem{Eldridge2012-hg}
Eldridge S, Kerry S. {\it A Practical Guide to Cluster Randomised Trials in
  Health Services Research}.
\newblock Statistics in Practice. Chichester, UK: John Wiley \& Sons;
\newblock 2012.

\bibitem{Hayes2017}
Hayes RJ, Moulton LH. {\it Cluster Randomised Trials}. 2nd~ed.
\newblock Chapman \& Hall/CRC Interdisciplinary Statistics Series. Boca Raton,
  FL: CRC Press;
\newblock 2017.

\bibitem{Kahn2018-ku}
Kahn R, Rid A, Smith PG, Eyal N, Lipsitch M. Choices in vaccine trial design in
  epidemics of emerging infections. {\it PLoS Med.} 2018\string; 15(8)\string:
  e1002632.

\bibitem{Hitchings2018-is}
Hitchings MDT, Lipsitch M, Wang R, Bellan SE. Competing Effects of Indirect
  Protection and Clustering on the Power of {Cluster-Randomized} Controlled
  Vaccine Trials. {\it Am. J. Epidemiol.} 2018\string; 187(8)\string:
  1763--1771.

\bibitem{Bellan2019-dk}
Bellan SE, Eggo RM, Gsell PS, et al. An online decision tree for vaccine
  efficacy trial design during infectious disease epidemics: The
  {InterVax-Tool}. {\it Vaccine} 2019\string; 37(31)\string: 4376--4381.

\bibitem{Brown2006-wu}
Brown CA, Lilford RJ. The stepped wedge trial design: a systematic review. {\it
  BMC Med. Res. Methodol.} 2006\string; 6(1)\string: 54.

\bibitem{Hemming2015-yq}
Hemming K, Haines TP, Chilton PJ, Girling AJ, Lilford RJ. The stepped wedge
  cluster randomised trial: rationale, design, analysis, and reporting. {\it
  BMJ} 2015\string; 350\string: h391.

\bibitem{Beard2015-kp}
Beard E, Lewis JJ, Copas A, et al. Stepped wedge randomised controlled trials:
  systematic review of studies published between 2010 and 2014. {\it Trials}
  2015\string; 16\string: 353.

\bibitem{Davey2015-no}
Davey C, Hargreaves J, Thompson JA, et al. Analysis and reporting of stepped
  wedge randomised controlled trials: synthesis and critical appraisal of
  published studies, 2010 to 2014. {\it Trials} 2015\string; 16\string: 358.

\bibitem{Barker2016-gn}
Barker D, McElduff P, D'Este C, Campbell MJ. Stepped wedge cluster randomised
  trials: a review of the statistical methodology used and available. {\it BMC
  Med. Res. Methodol.} 2016\string; 16(1)\string: 69.

\bibitem{Hussey2007-ay}
Hussey MA, Hughes JP. Design and analysis of stepped wedge cluster randomized
  trials. {\it Contemp. Clin. Trials} 2007\string; 28(2)\string: 182--191.

\bibitem{Copas2015-so}
Copas AJ, Lewis JJ, Thompson JA, Davey C, Baio G, Hargreaves JR. Designing a
  stepped wedge trial: three main designs, carry-over effects and randomisation
  approaches. {\it Trials} 2015\string; 16\string: 352.

\bibitem{Prost2015-ug}
Prost A, Binik A, Abubakar I, et al. Logistic, ethical, and political
  dimensions of stepped wedge trials: critical review and case studies. {\it
  Trials} 2015\string; 16\string: 351.

\bibitem{Tugwell2019-nn}
Tugwell P, Knottnerus JA. Stepped wedge designs are coming of age in clinical
  epidemiology. {\it J. Clin. Epidemiol.} 2019\string; 107\string: vi--viii.

\bibitem{Mdege2011-mc}
Mdege ND, Man MS, Taylor CA, Torgerson DJ. Systematic review of stepped wedge
  cluster randomized trials shows that design is particularly used to evaluate
  interventions during routine implementation. {\it J. Clin. Epidemiol.}
  2011\string; 64(9)\string: 936--948.

\bibitem{Keriel-Gascou2014-yh}
Keriel-Gascou M, Buchet-Poyau K, Rabilloud M, Duclos A, Colin C. A stepped
  wedge cluster randomized trial is preferable for assessing complex health
  interventions. {\it J. Clin. Epidemiol.} 2014\string; 67(7)\string: 831--833.

\bibitem{Kotz2012-rc}
Kotz D, Spigt M, Arts ICW, Crutzen R, Viechtbauer W. Use of the stepped wedge
  design cannot be recommended: a critical appraisal and comparison with the
  classic cluster randomized controlled trial design. {\it J. Clin. Epidemiol.}
  2012\string; 65(12)\string: 1249--1252.

\bibitem{Hargreaves2015-zp}
Hargreaves JR, Copas AJ, Beard E, et al. Five questions to consider before
  conducting a stepped wedge trial. {\it Trials} 2015\string; 16\string: 350.

\bibitem{Eyal2017-gf}
Eyal N, Lipsitch M. Vaccine testing for emerging infections: the case for
  individual randomisation. {\it J. Med. Ethics} 2017\string; 43(9)\string:
  625--631.

\bibitem{Mdege2014-pm}
Mdege ND, Kanaan M. Response to {Keriel-Gascou} et al. {A}ddressing assumptions
  on the stepped wedge randomized trial design. {\it J. Clin. Epidemiol.}
  2014\string; 67(7)\string: 833--834.

\bibitem{Viechtbauer2014-ox}
Viechtbauer W, Kotz D, Spigt M, Arts ICW, Crutzen R. Response to
  {Keriel-Gascou} et al.: {h}igher efficiency and other alleged advantages are
  not inherent to the stepped wedge design. {\it J. Clin. Epidemiol.}
  2014\string; 67(7)\string: 834--836.

\bibitem{Nickless2018-cm}
Nickless A, Voysey M, Geddes J, Yu LM, Fanshawe TR. Mixed effects approach to
  the analysis of the stepped wedge cluster randomised trial---Investigating
  the confounding effect of time through simulation. {\it PLoS One}
  2018\string; 13(12)\string: e0208876.

\bibitem{Thompson2017-or}
Thompson JA, Fielding KL, Davey C, Aiken AM, Hargreaves JR, Hayes RJ. Bias and
  inference from misspecified mixed-effect models in stepped wedge trial
  analysis. {\it Stat. Med.} 2017\string; 36(23)\string: 3670--3682.

\bibitem{Hooper2016-ih}
Hooper R, Teerenstra S, Hoop dE, Eldridge S. Sample size calculation for
  stepped wedge and other longitudinal cluster randomised trials. {\it Stat.
  Med.} 2016\string; 35(26)\string: 4718--4728.

\bibitem{Hartford2000-zz}
Hartford A, Davidian M. Consequences of misspecifying assumptions in nonlinear
  mixed effects models. {\it Comput. Stat. Data Anal.} 2000\string;
  34(2)\string: 139--164.

\bibitem{Heagerty2001-vg}
Heagerty PJ, Kurland BF. Misspecified maximum likelihood estimates and
  generalised linear mixed models. {\it Biometrika} 2001\string; 88(4)\string:
  973--985.

\bibitem{Agresti2004-qa}
Agresti A, Caffo B, Ohman-Strickland P. Examples in which misspecification of a
  random effects distribution reduces efficiency, and possible remedies. {\it
  Comput. Stat. Data Anal.} 2004\string; 47(3)\string: 639--653.

\bibitem{Litiere2007-an}
Liti{\`e}re S, Alonso A, Molenberghs G. Type {I} and type {II} error under
  random-effects misspecification in generalized linear mixed models. {\it
  Biometrics} 2007\string; 63(4)\string: 1038--1044.

\bibitem{Litiere2008-gw}
Liti{\`e}re S, Alonso A, Molenberghs G. The impact of a misspecified
  random-effects distribution on the estimation and the performance of
  inferential procedures in generalized linear mixed models. {\it Stat. Med.}
  2008\string; 27(16)\string: 3125--3144.

\bibitem{McCulloch2011-jz}
McCulloch CE, Neuhaus JM. Prediction of random effects in linear and
  generalized linear models under model misspecification. {\it Biometrics}
  2011\string; 67(1)\string: 270--279.

\bibitem{Wang2017-gr}
Wang R, De~Gruttola V. The use of permutation tests for the analysis of
  parallel and stepped-wedge cluster-randomized trials. {\it Stat. Med.}
  2017\string; 36(18)\string: 2831--2843.

\bibitem{Ji2017-cq}
Ji X, Fink G, Robyn PJ, Small DS. Randomization inference for stepped-wedge
  cluster-randomized trials: an application to community-based health
  insurance. {\it Ann. Appl. Stat.} 2017\string; 11(1)\string: 1--20.

\bibitem{Davidian1993-gb}
Davidian M, Gallant AR. The Nonlinear Mixed Effects Model with a Smooth Random
  Effects Density. {\it Biometrika} 1993\string; 80(3)\string: 475--488.

\bibitem{Zhang2001-gg}
Zhang D, Davidian M. Linear mixed models with flexible distributions of random
  effects for longitudinal data. {\it Biometrics} 2001\string; 57(3)\string:
  795--802.

\bibitem{Chen2002-zy}
Chen J, Zhang D, Davidian M. A {M}onte {C}arlo {EM} algorithm for generalized
  linear mixed models with flexible random effects distribution. {\it
  Biostatistics} 2002\string; 3(3)\string: 347--360.

\bibitem{Scott2017-wi}
Scott JM, deCamp A, Juraska M, Fay MP, Gilbert PB. Finite-sample corrected
  generalized estimating equation of population average treatment effects in
  stepped wedge cluster randomized trials. {\it Stat. Methods Med. Res.}
  2017\string; 26(2)\string: 583--597.

\bibitem{Thompson2018-so}
Thompson JA, Davey C, Fielding K, Hargreaves JR, Hayes RJ. Robust analysis of
  stepped wedge trials using cluster-level summaries within periods. {\it Stat.
  Med.} 2018\string; 37(16)\string: 2487--2500.

\bibitem{Hughes2019-ud}
Hughes JP, Heagerty PJ, Xia F, Ren Y. Robust Inference for the Stepped Wedge
  Design. {\it Biometrics} 2019.

\bibitem{Abadie2003-tp}
Abadie A, Gardeazabal J. The economic costs of conflict: a case study of the
  {B}asque {C}ountry. {\it Am. Econ. Rev.} 2003\string; 93(1)\string: 113--132.

\bibitem{Abadie2010-zy}
Abadie A, Diamond A, Hainmueller J. Synthetic control methods for comparative
  case studies: estimating the effect of {C}alifornia's tobacco control
  program. {\it J. Am. Stat. Assoc.} 2010\string; 105(490)\string: 493--505.

\bibitem{Abadie2015-qi}
Abadie A, Diamond A, Hainmueller J. Comparative Politics and the Synthetic
  Control Method. {\it Am. J. Pol. Sci.} 2015\string; 59(2)\string: 495--510.

\bibitem{Trajman2015-mj}
Trajman A, Durovni B, Saraceni V, et al. Impact on patients' treatment outcomes
  of {XpertMTB/RIF} implementation for the diagnosis of tuberculosis:
  follow-up of a stepped-wedge randomized clinical trial. {\it PLoS One}
  2015\string; 10(4)\string: e0123252.

\bibitem{Gautier2009-pn}
Gautier PA, Siegmann A, Van~Vuuren A. Terrorism and attitudes towards
  minorities: the effect of the {Theo} van {Gogh} murder on house prices in
  {Amsterdam}. {\it J. Urban Econ.} 2009\string; 65(2)\string: 113--126.

\bibitem{Abadie2011-ho}
Abadie A, Diamond A, Hainmueller J. Synth: an {R} package for synthetic control
  methods in comparative case studies. {\it Journal of Statistical Software}
  2011\string; 42(13)\string: 1--17.

\bibitem{Dube2015-pl}
Dube A, Zipperer B. Pooling multiple case studies using synthetic controls: an
  application to minimum wage policies. Tech. Rep. 8944, Institution for the
  Study of Labor;  2015.

\bibitem{Donohue2018-qb}
Donohue JJ, Aneja A, Weber KD. Right-to-carry laws and violent crime: a
  comprehensive assessment using panel data and a state-level synthetic control
  analysis. Tech. Rep. 23510, National Bureau of Economic Research;  2018.

\bibitem{Powell2017-qf}
Powell D. Synthetic control estimation beyond case studies: does the minimum
  wage reduce employment?. Tech. Rep. WR-1142, RAND Labor \& Population;  2017.

\bibitem{Everitt1995-cj}
Everitt BS. The analysis of repeated measures: a practical review with
  examples. {\it J. R. Stat. Soc. Ser. D The Statistician} 1995\string;
  44(1)\string: 113--135.

\bibitem{Jones1996-by}
Jones B, Donev AN. Modelling and design of cross-over trials. {\it Stat. Med.}
  1996\string; 15(13)\string: 1435--1446.

\bibitem{Omar1999-kf}
Omar RZ, Wright EM, Turner RM, Thompson SG. Analysing repeated measurements
  data: a practical comparison of methods. {\it Stat. Med.} 1999\string;
  18(13)\string: 1587--1603.

\bibitem{Fitzmaurice2011-ab}
Fitzmaurice GM. {\it Applied longitudinal analysis}. 2nd~ed.
\newblock Wiley series in probability and statistics. Hoboken, N.J.: Wiley;
\newblock 2011.

\bibitem{Durovni2014-nq}
Durovni B, Saraceni V, Hof v.~dS, et al. Impact of replacing smear microscopy
  with {XpertMTB/RIF} for diagnosing tuberculosis in {Brazil}: a stepped-wedge
  cluster-randomized trial. {\it PLoS Med.} 2014\string; 11(12)\string:
  e1001766.

\bibitem{Alonso2008-nm}
Alonso A, Liti{\`e}re S, Molenberghs G. A family of tests to detect
  misspecifications in the random-effects structure of generalized linear mixed
  models. {\it Comput. Stat. Data Anal.} 2008\string; 52(9)\string: 4474--4486.

\bibitem{Huang2009-vi}
Huang X. Diagnosis of random-effect model misspecification in generalized
  linear mixed models for binary response. {\it Biometrics} 2009\string;
  65(2)\string: 361--368.

\bibitem{Verbeke2013-oq}
Verbeke G, Molenberghs G. The gradient function as an exploratory
  goodness-of-fit assessment of the random-effects distribution in mixed
  models. {\it Biostatistics} 2013\string; 14(3)\string: 477--490.

\bibitem{Meintanis2016-al}
Meintanis SG, Allison JS, Santana L. Diagnostic tests for the distribution of
  random effects in multivariate mixed effects models. {\it Commun. Stat.
  Theory Methods} 2016\string; 45(1)\string: 201--215.

\bibitem{Drikvandi2017-rf}
Drikvandi R. Nonlinear mixed-effects models for pharmacokinetic data analysis:
  assessment of the random-effects distribution. {\it J. Pharmacokinet.
  Pharmacodyn.} 2017\string; 44(3)\string: 223--232.

\bibitem{Singer2013-op}
Singer JM, Nobre JS, Rocha FMM. Diagnostic and treatment for linear mixed
  models. In: Proceedings of the 59th {ISI} World Statistics Congress, August;
  2013\string: 25--30.

\bibitem{Ritz2004-yq}
Ritz C. Goodness-of-fit tests for mixed models. {\it Scand. Stat. Theory Appl.}
  2004\string; 31(3)\string: 443--458.

\bibitem{Yap2011-bu}
Yap BW, Sim CH. Comparisons of various types of normality tests. {\it J. Stat.
  Comput. Simul.} 2011\string; 81(12)\string: 2141--2155.

\bibitem{Hubbard2010-ly}
Hubbard AE, Ahern J, Fleischer NL, et al. To {GEE} or not to {GEE}: comparing
  population average and mixed models for estimating the associations between
  neighborhood risk factors and health. {\it Epidemiology} 2010\string;
  21(4)\string: 467--474.

\bibitem{Botosaru2019-rg}
Botosaru I, Ferman B. On the role of covariates in the synthetic control
  method. {\it Econom. J.} 2019\string; 22.

\bibitem{Bruhn2017-on}
Bruhn CAW, Hetterich S, Schuck-Paim C, et al. Estimating the population-level
  impact of vaccines using synthetic controls. {\it Proc. Natl. Acad. Sci. U.
  S. A.} 2017\string; 114(7)\string: 1524--1529.

\bibitem{Kim2019-gp}
Kim S, Lee C, Gupta S. Bayesian synthetic control methods. Tech. Rep. 3382359,
  SSRN;  2019.

\bibitem{Cerulli2019-mh}
Cerulli G. A flexible synthetic control method for modeling policy evaluation.
  {\it Econ. Lett.} 2019\string; 182\string: 40--44.

\bibitem{Xu2017-ne}
Xu Y. Generalized synthetic control method: causal inference with interactive
  fixed effects models. {\it Polit. Anal.} 2017\string; 25(1)\string: 57--76.

\bibitem{Ben-Michael2018-bq}
Ben-Michael E, Feller A, Rothstein J. The augmented synthetic control method.
  {\it arXiv:1811.04170 [stat.ME]} 2018.

\bibitem{Hernan2016-aj}
Hern{\'a}n MA, Sauer BC, Hern{\'a}ndez-D{\'\i}az S, Platt R, Shrier I.
  Specifying a target trial prevents immortal time bias and other
  self-inflicted injuries in observational analyses. {\it J. Clin. Epidemiol.}
  2016\string; 79\string: 70--75.

\bibitem{Garcia-Albeniz2017-wk}
Garc{\'\i}a-Alb{\'e}niz X, Hsu J, Hern{\'a}n MA. The value of explicitly
  emulating a target trial when using real world evidence: an application to
  colorectal cancer screening. {\it Eur. J. Epidemiol.} 2017\string;
  32(6)\string: 495--500.

\end{thebibliography}

\clearpage

\section*{Figures}


\begin{center}
    \begin{figure}[!ht]
    \centering
    \includegraphics[width=\textwidth]{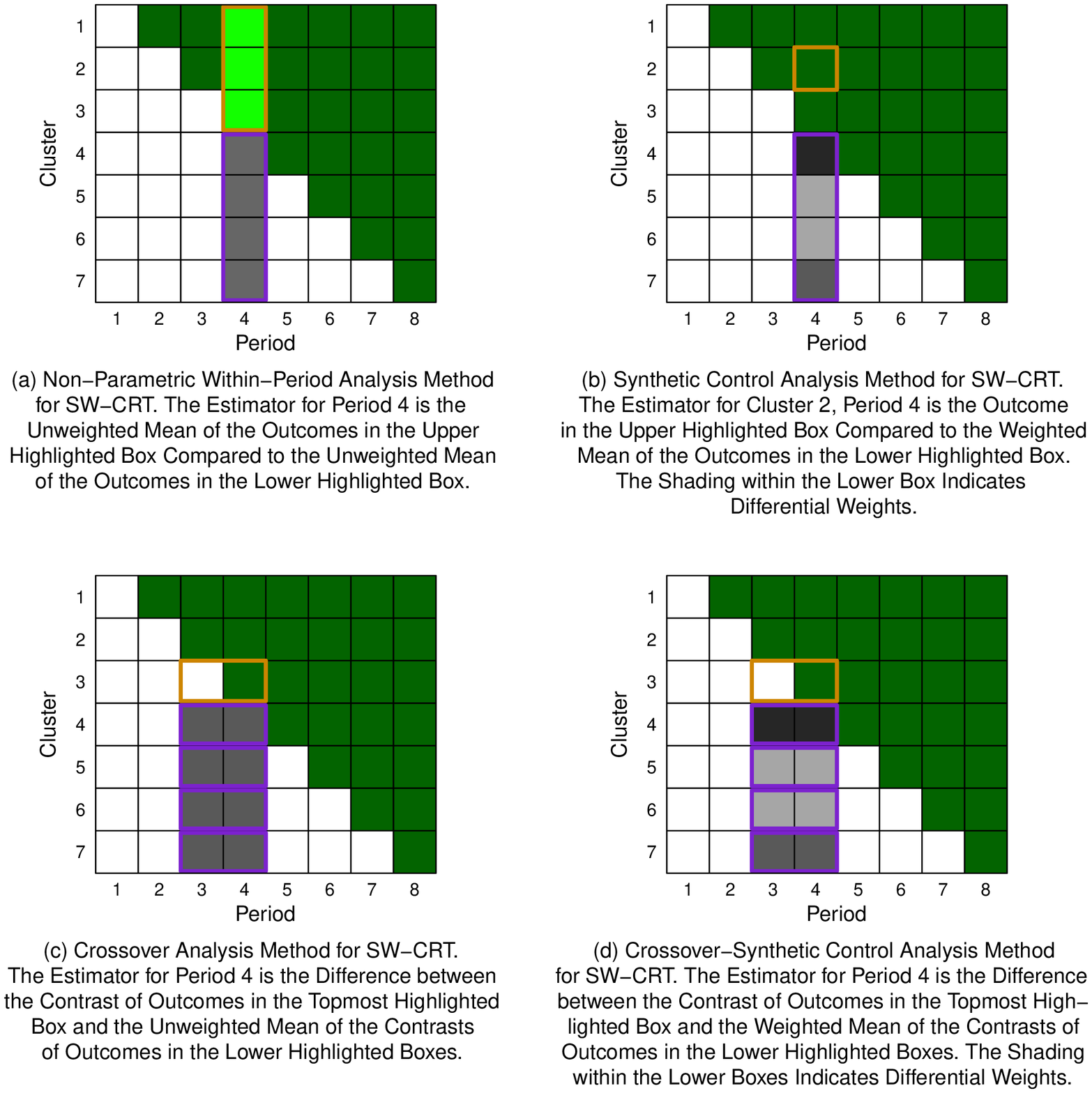}
    \caption{Schematic Representations of Several Existing and Novel Analysis Methods for a SW-CRT with Seven Clusters, Eight Periods, and One Crossover Per Period. Dark Green Boxes Indicate Cluster-Periods on Intervention and White/Gray Boxes Indicate Cluster-Periods on Control.}\label{Fig1}
    \end{figure}
\end{center}

\clearpage


\begin{center}
    \begin{figure}[!ht]
    \centering
    \includegraphics[width=\textwidth]{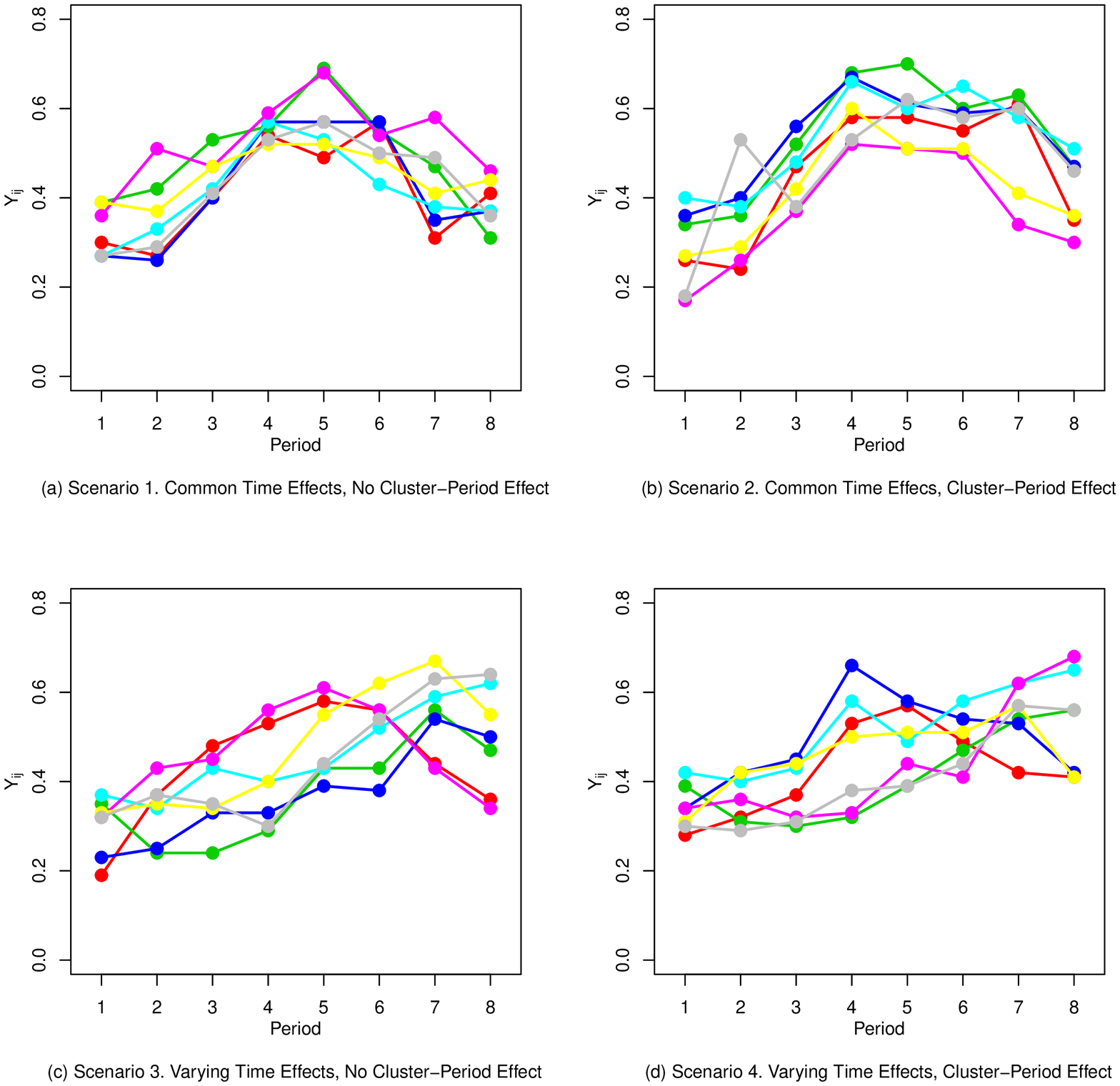}
    \caption{Sample Generated Data for Four Risk Difference Scenarios (Simulation 1) with No Treatment Effect. Each Line Represents the Simulated Cluster-Level Outcome for One Cluster over Eight Time Periods.}\label{Fig2}
\end{figure}
\end{center}

\clearpage

\begin{center}
    \begin{figure}[!ht]
    \centering
    \includegraphics[width=.87\textwidth]{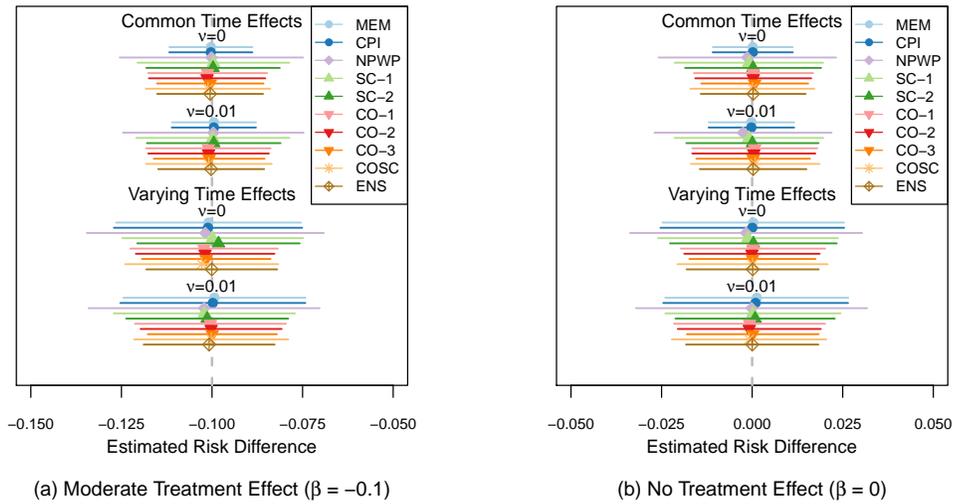}
    \caption{Mean Treatment Effect Estimates and 1/2-Standard Deviation of Estimates across 1,000 Simulations for Risk Difference Scenarios (Simulation 1) by Analysis Method}\label{Fig3}
\end{figure}
\end{center}

\begin{center}
    \begin{figure}[!ht]
    \centering
    \includegraphics[width=.87\textwidth]{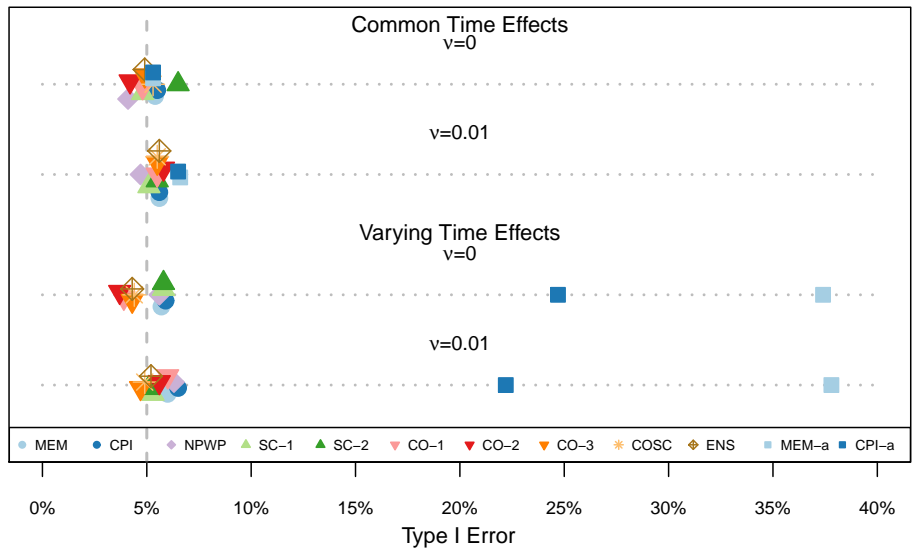}
    \caption{Type I Error Rate across 1,000 Simulations for Risk Difference Scenarios (Simulation 1) by Analysis Method}\label{Fig4}
\end{figure}
\end{center}

\clearpage

\begin{center}
    \begin{figure}[!ht]
    \centering
    \includegraphics[width=.87\textwidth]{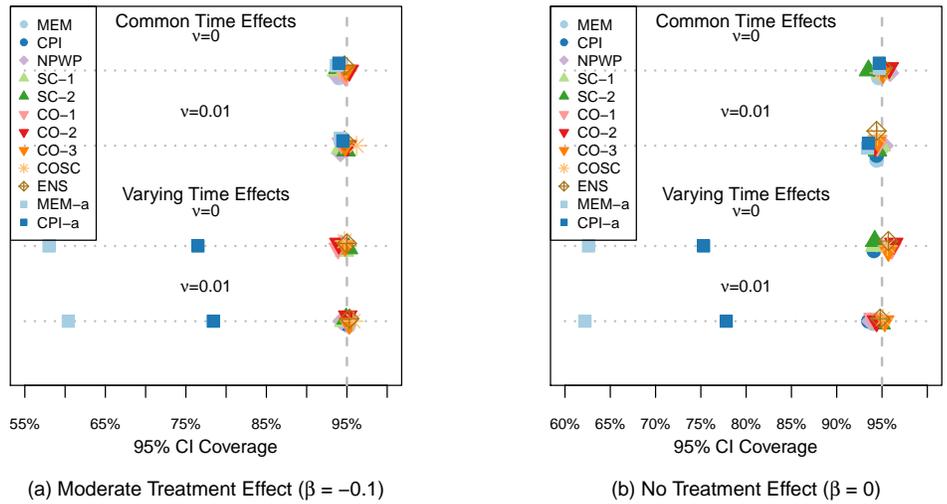}
    \caption{95\% Confidence Interval Coverage Rate across 1,000 Simulations for Risk Difference Scenarios (Simulation 1) by Analysis Method}\label{Fig5}
\end{figure}
\end{center}

\begin{center}
    \begin{figure}[!ht]
    \centering
    \includegraphics[width=.87\textwidth]{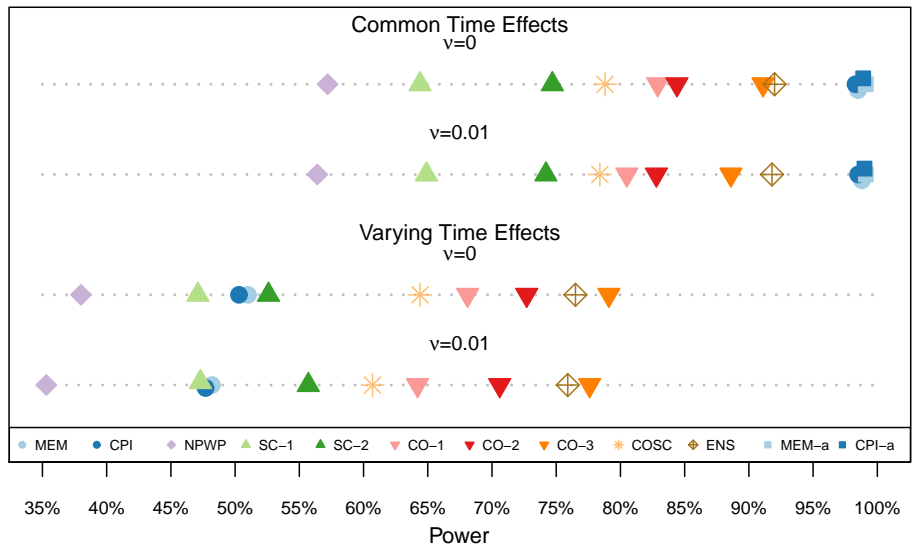}
    \caption{Power across 1,000 Simulations for Risk Difference Scenarios (Simulation 1) with Moderate Treatment Effect ($\beta = -0.1$) by Analysis Method}\label{Fig6}
\end{figure}
\end{center}

\clearpage


\begin{center}
    \begin{figure}[!ht]
    \centering
    \includegraphics[width=\textwidth]{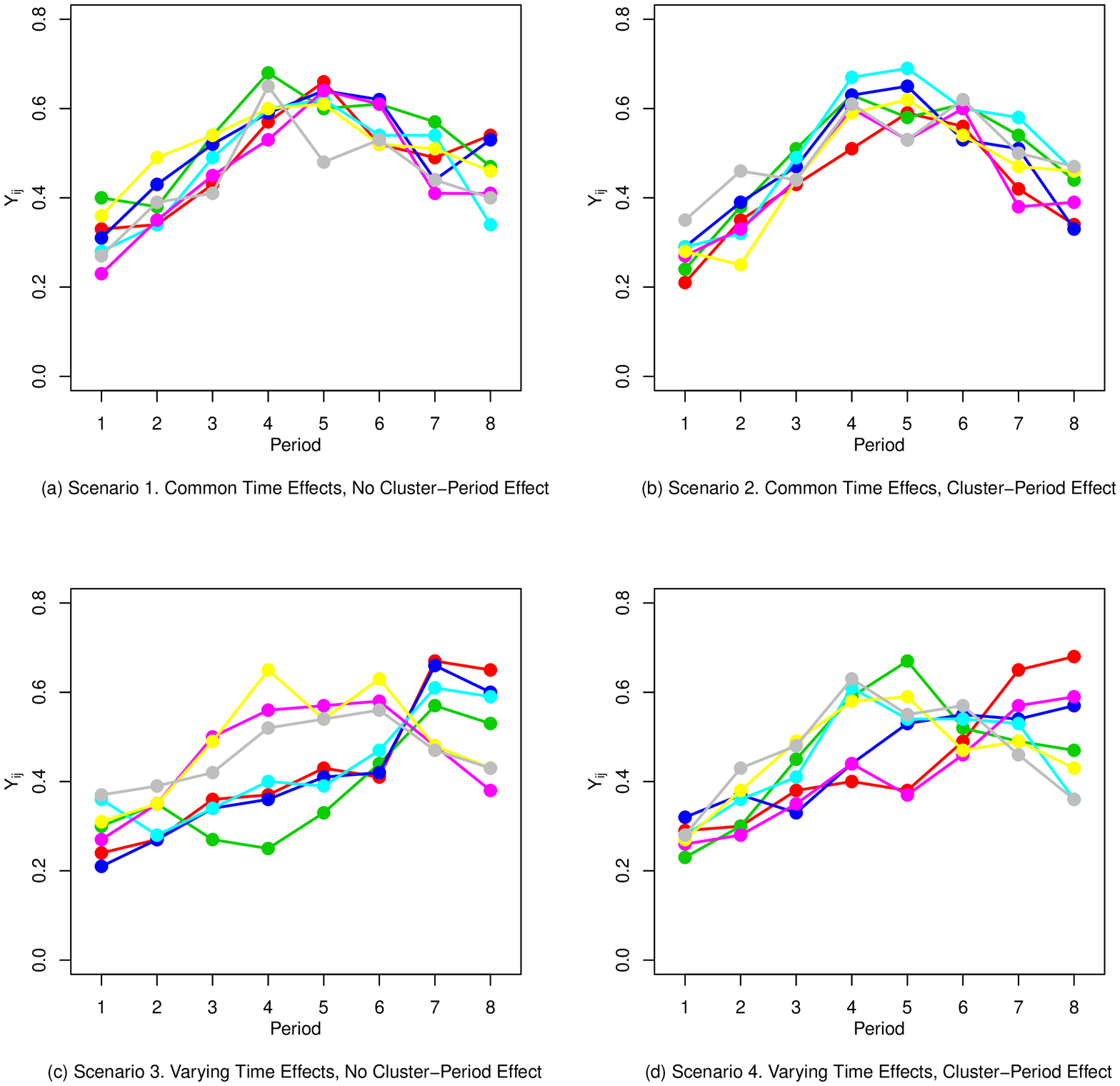}
    \caption{Sample Generated Data for Four Odds Ratio Scenarios (Simulation 2) with No Treatment Effect. Each Line Represents the Simulated Cluster-Level Outcome for One Cluster over Eight Time Periods.}\label{Fig7}
\end{figure}
\end{center}

\clearpage

\begin{center}
    \begin{figure}[!ht]
    \centering
    \includegraphics[width=.87\textwidth]{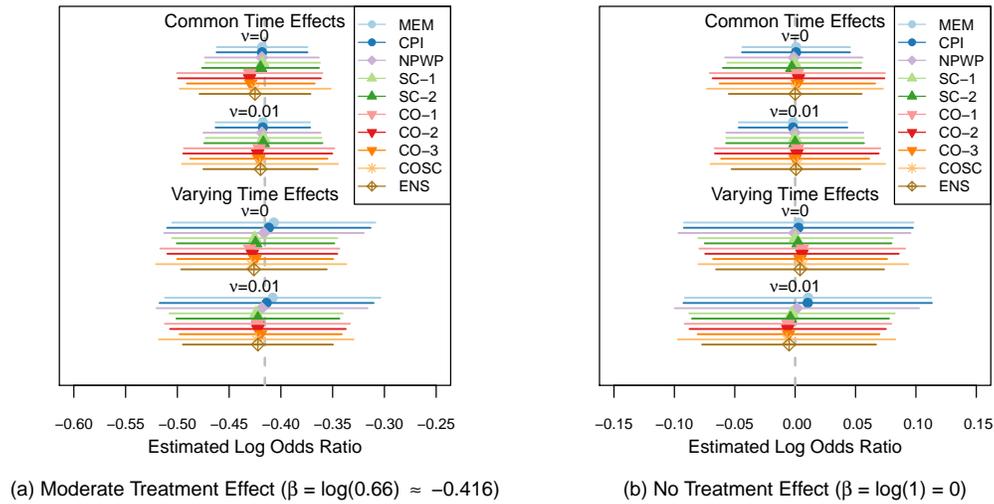}
    \caption{Mean Treatment Effect Estimates and 1/2-Standard Deviation of Estimates across 1,000 Simulations for Odds Ratio Scenarios (Simulation 2) by Analysis Method}\label{Fig8}
\end{figure}
\end{center}

\begin{center}
    \begin{figure}[!ht]
    \centering
    \includegraphics[width=.87\textwidth]{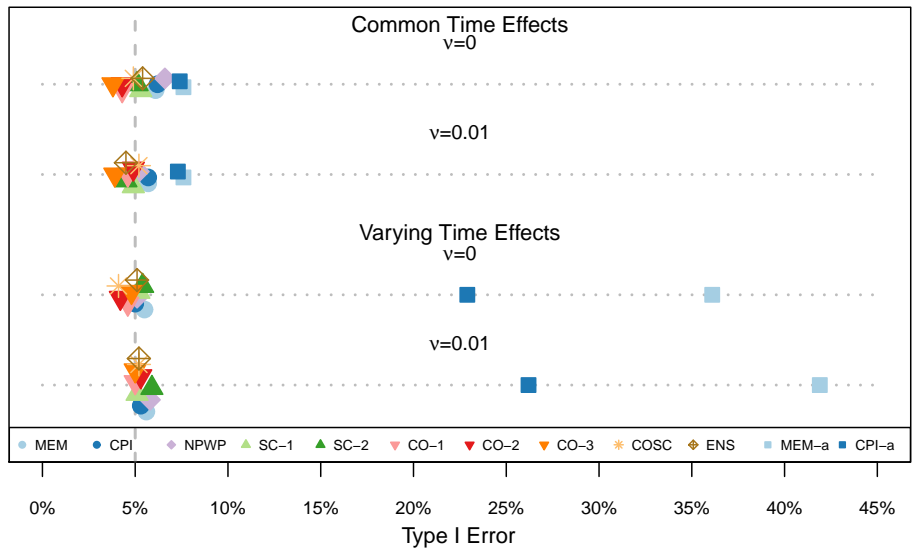}
    \caption{Type I Error Rate across 1,000 Simulations for Odds Ratio Scenarios (Simulation 2) by Analysis Method}\label{Fig9}
\end{figure}
\end{center}

\clearpage

\begin{center}
    \begin{figure}[!ht]
    \centering
    \includegraphics[width=.87\textwidth]{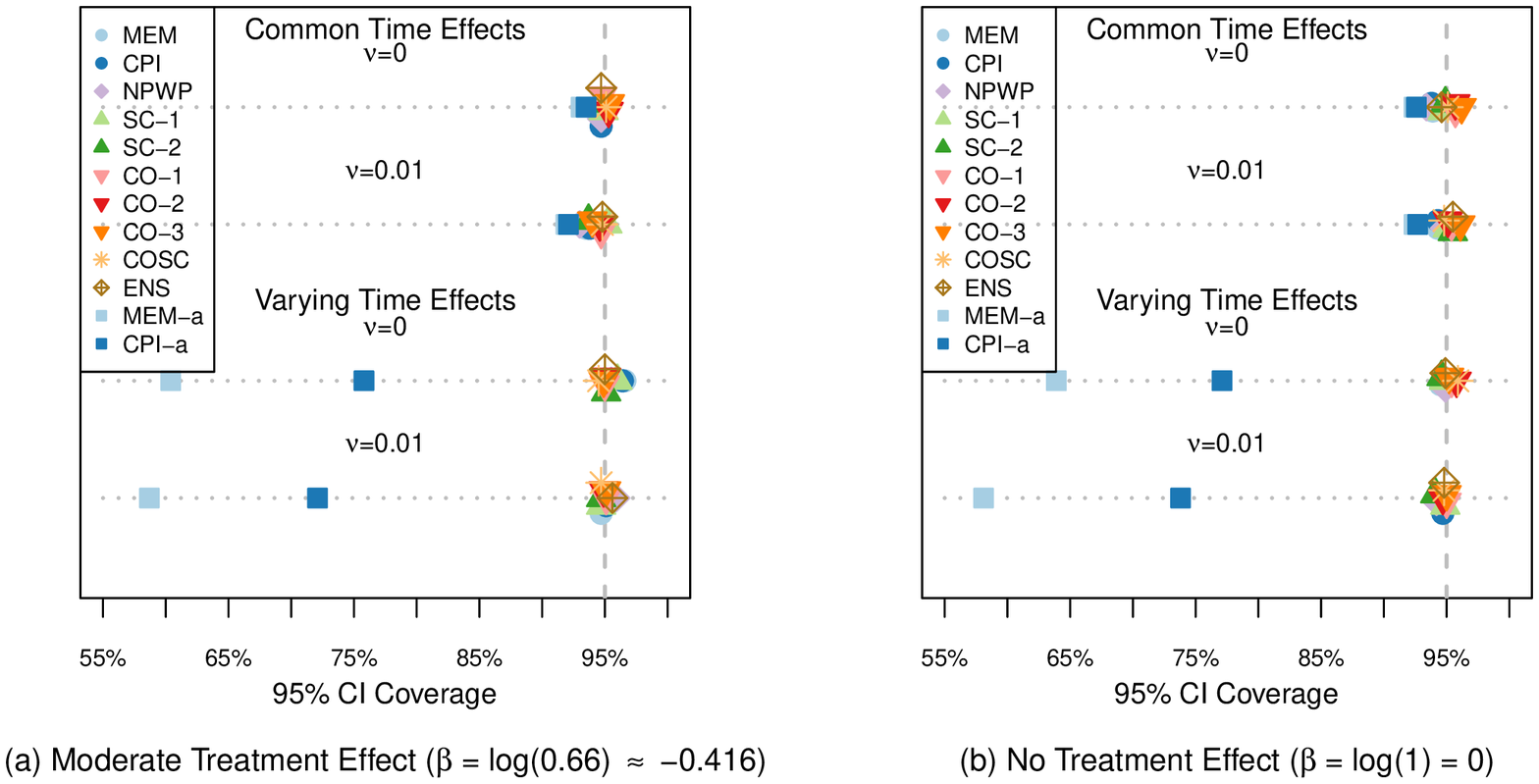}
    \caption{95\% Confidence Interval Coverage Rate across 1,000 Simulations for Odds Ratio Scenarios (Simulation 2) by Analysis Method}\label{Fig10}
\end{figure}
\end{center}

\begin{center}
    \begin{figure}[!ht]
    \centering
    \includegraphics[width=.87\textwidth]{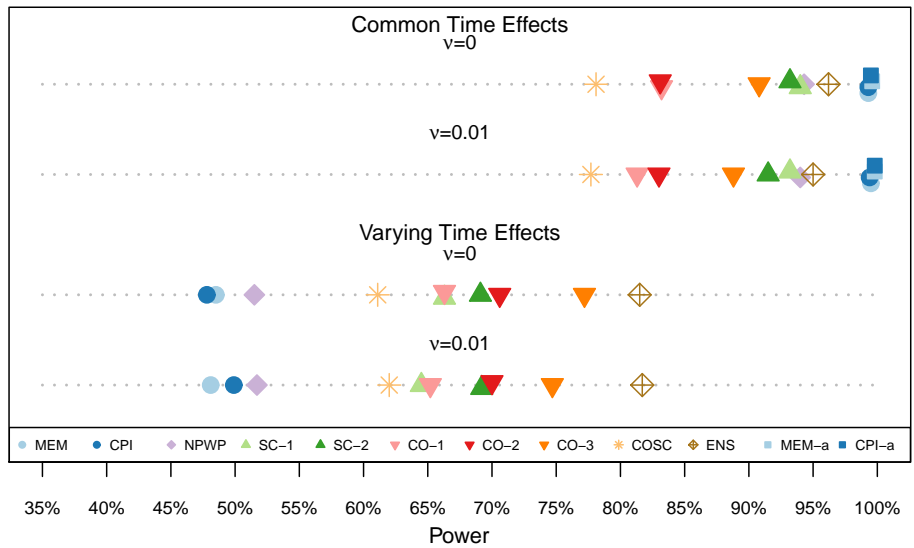}
    \caption{Power across 1,000 Simulations for Odds Ratio Scenarios (Simulation 2) with Moderate Treatment Effect ($\beta = \log(0.66) \approx -0.416$) by Analysis Method}\label{Fig11}
\end{figure}
\end{center}

\clearpage


\begin{center}
    \begin{figure}[!ht]
    \centering
    \includegraphics[width=\textwidth]{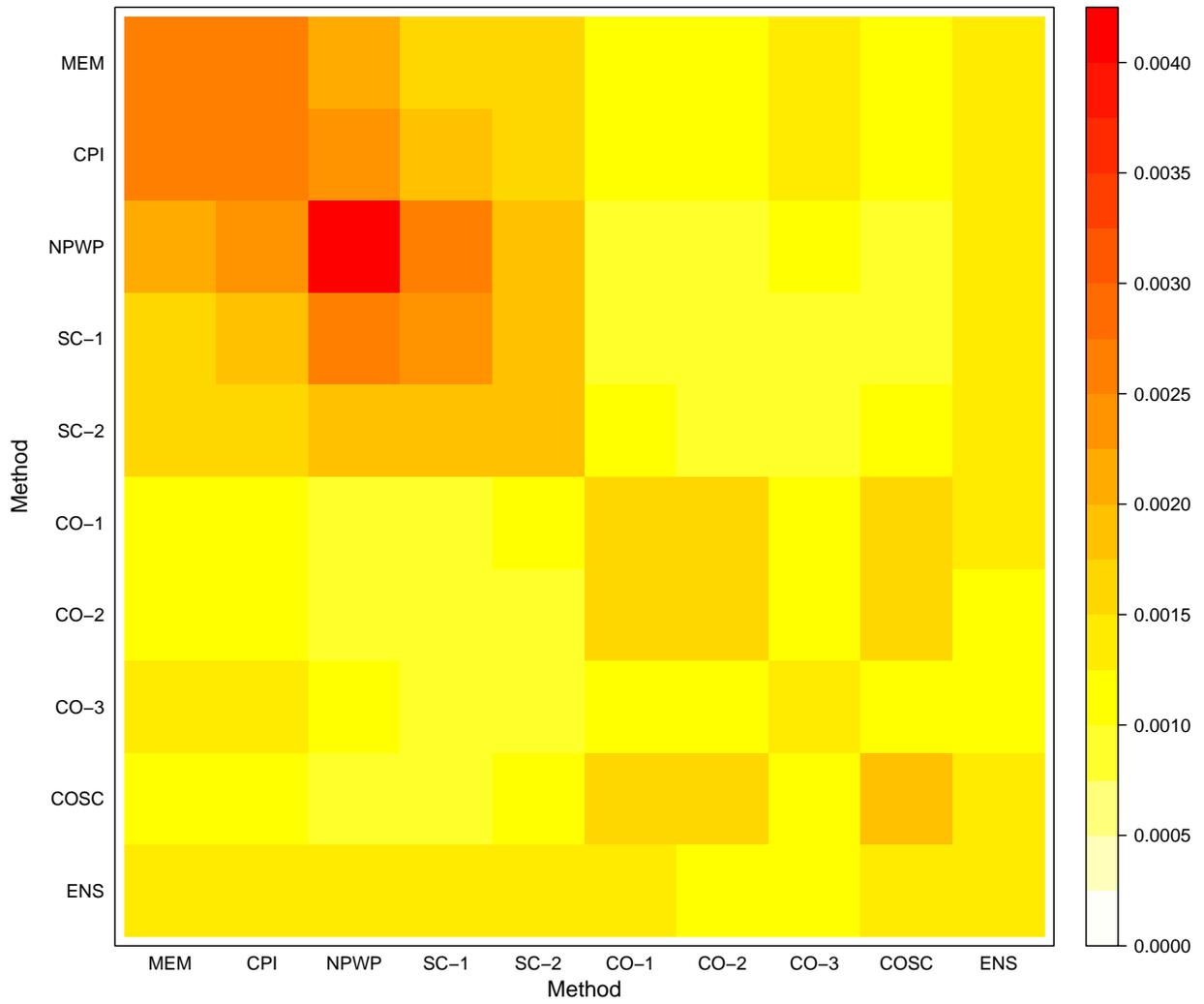}
    \caption{Pairwise Covariance between Effect Estimates from Different Methods: Simulation 1 (Risk Difference), Scenario 4, No Treatment Effect}\label{Fig12}
    \end{figure}
\end{center}

\clearpage

\begin{center}
    \begin{figure}[!ht]
    \centering
    \includegraphics[width=\textwidth]{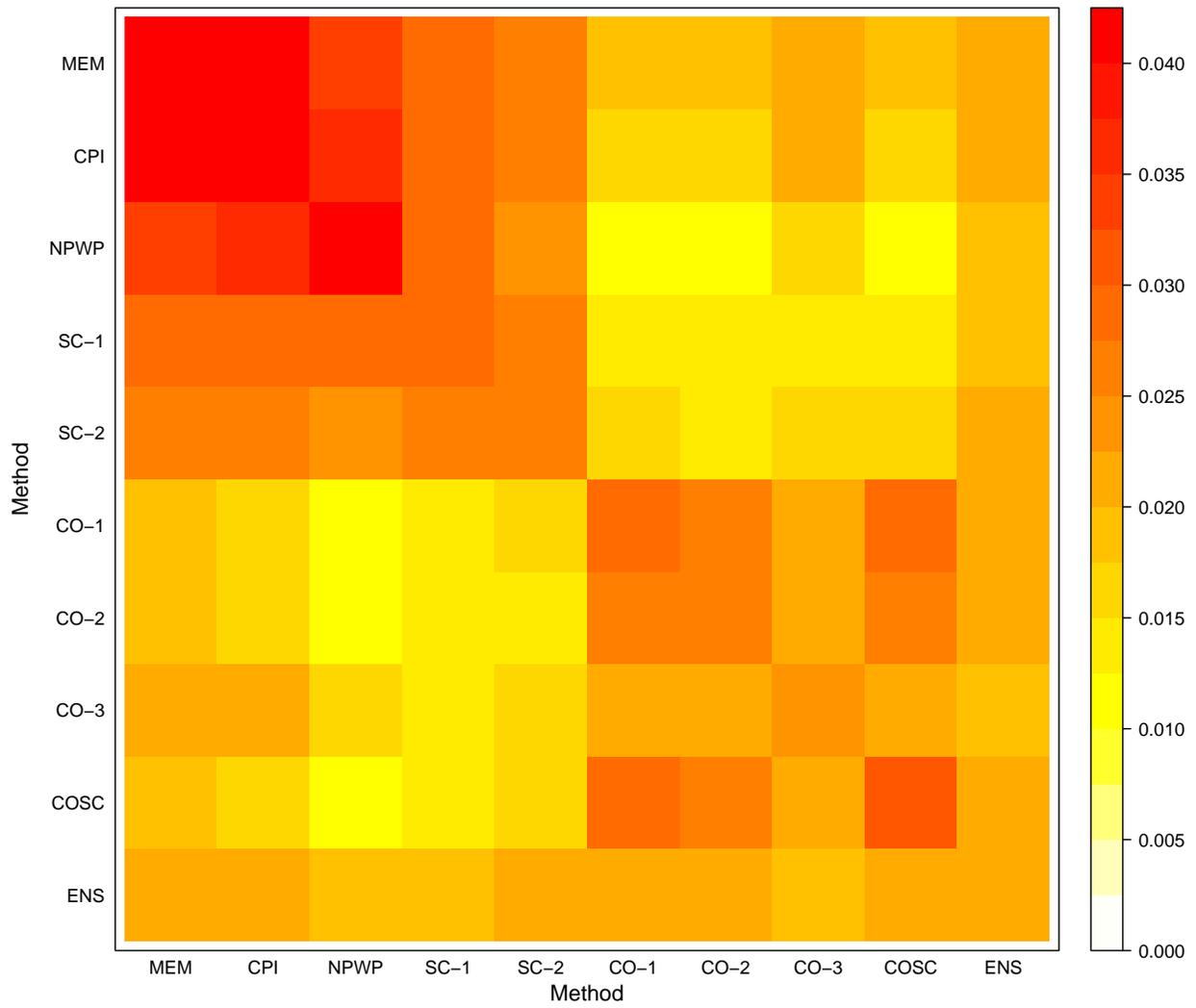}
    \caption{Pairwise Covariance between Effect Estimates from Different Methods: Simulation 2 (Odds Ratio), Scenario 4, No Treatment Effect}\label{Fig13}
    \end{figure}
\end{center}

\clearpage


\appendix
\appendixpage
\numberwithin{equation}{section}
\numberwithin{figure}{section}

\section{Proofs of theorems} \label{app1}
\setcounter{equation}{0}
\renewcommand{\theequation}{A\arabic{equation}}

\begin{theorem} \label{SCThm1}
Suppose that for each cluster $i$, denoting by $j_i$ the last period for which cluster $i$ is on control, $E[(Y_{i,1},Y_{i,2},\ldots,Y_{i,j_i})] = (Y_{\cdot 1},Y_{\cdot 2},\ldots,Y_{\cdot j_i}) \equiv \bm{Y}_{\cdot j_i}$ and that the distribution of $(Y_{i,1},Y_{i,2},\ldots,Y_{i,j_i})$ is symmetric about $\bm{Y}_{\cdot j_i}$. Suppose further that the cluster-level outcomes from two different clusters are uncorrelated conditional on the full vector of expected outcomes, $\bm{Y}_{\cdot J}$, and the treatment effect $\beta$. Then, for any cluster $i^*$ in any period $j^*$ where that cluster is on intervention ($j^* > j_{i^*}$), the synthetic control estimator $Z_{i^*,j^*}$ is an unbiased estimate of the marginal (across clusters) expectation for an untreated cluster in period $j^*$. That is, $E[Z_{i^*,j^*}] = Y_{\cdot j^*}$.
\end{theorem}

\begin{proof}
Consider a target cluster $i^*$ and period $j^*$ such that $X_{i^*,j^*} = 1$. Let $m_1,\ldots,m_{n^*}$ index the $n^* \equiv n_{0,j^*}$ clusters on control (``donor clusters'') in period $j^*$. For any cluster $i$, define $\bm{Y}_{i,j_{i^*}} \equiv (Y_{i,1},\ldots,Y_{i,j_{i^*}})^T$, where $j_{i^*}$ is the last period for which cluster $i^*$ is on control (and thus $j_{i^*} < j^*$). Denote the $j_{i^*} \times n^*$ matrix of pre-intervention donor cluster outcomes by $\bm{Y} \equiv \begin{pmatrix} \bm{Y}_{m_1,j_{i^*}}, & \hdots, & \bm{Y}_{m_{n^*},j_{i^*}} \end{pmatrix}$. Construct a $j_{i^*} \times n^*$ matrix of pre-intervention target cluster outcomes by repeating the vector $\bm{Y}_{i^*,j_{i^*}}$ $n^*$ times: $\bm{Y}_{i^*} \equiv \begin{pmatrix} \bm{Y}_{i^*,j_{i^*}}, & \hdots, & \bm{Y}_{i^*,j_{i^*}} \end{pmatrix}$.

By definition of the synthetic control estimator, $Z_{i^*,j^*} = \sum_{n=1}^{n^*} v_{i^*,j^*,n} Y_{m_n,j^*}$. The vector of these weights is denoted $\bm{v}_{i^*,j^*}$ and lies in the set:
\[ \mathcal{V} \equiv \left\{\bm{v} \in \mathbb{R}^{n}: ~ \sum_{n=1}^{n^*} v_n = 1 ~ \cap ~ 0 \le v_n \le 1 ~ \forall ~ 1 \le n \le n^* \right\}. \]
Note that for all $\bm{v} \in \mathcal{V}$, $\bm{Y}_{i^*} \bm{v} = \bm{Y}_{i^*,j_{i^*}}$. Then we can write:
\begin{equation}
    \bm{v}_{i^*,j^*} = \argmin_{\bm{v} \in \mathcal{V}} \norm*{\bm{Y}_{i^*,j_{i^*}} - \bm{Y} \bm{v}} = \argmin_{\bm{v} \in \mathcal{V}} \norm*{(\bm{Y}_{i^*} - \bm{Y}) \bm{v}} \equiv \argmin_{\bm{v} \in \mathcal{V}} \norm*{\bm{Y}_{i^*,j^*}^{diff} \bm{v}}, \label{vijDefn}
\end{equation}
where the difference matrix is:
\begin{align*}
    \bm{Y}_{i^*,j^*}^{diff} &\equiv \bm{Y}_{i^*} - \bm{Y} = \begin{pmatrix} Y_{i^*,1} - Y_{m_1,1} & Y_{i^*,1} - Y_{m_2,1} & \hdots & Y_{i^*,1} - Y_{m_{n^*},1} \\ \vdots & \vdots & \ddots & \vdots \\ Y_{i^*,j_{i^*}} - Y_{m_1,j_{i^*}} & Y_{i^*,j_{i^*}} - Y_{m_2,j_{i^*}} & \hdots & Y_{i^*,j_{i^*}} - Y_{m_{n^*},j_{i^*}} \end{pmatrix} \\
    &= \begin{pmatrix} \bm{Y}_{i^*,j_{i^*}} - \bm{Y}_{m_1,j_{i^*}} & \hdots & \bm{Y}_{i^*,j_{i^*}}  - \bm{Y}_{m_{n^*},j_{i^*}} \end{pmatrix}
\end{align*}
By the symmetry and independence assumptions, for any $n=1,\ldots,n^*$, $\bm{Y}_{i^*,j_{i^*}}$ and $\bm{Y}_{m_n,j_{i^*}}$ are independent and both are symmetrically distributed with a common mean $\bm{Y}_{\cdot j_{i^*}}$. Thus, each column of $\bm{Y}_{i^*,j^*}^{diff}$ is symmetrically distributed with expectation $\bm{0}$ and hence the matrix $\bm{Y}_{i^*,j^*}^{diff}$ is symmetrically distributed with expectation $\bm{0}$.

Moreover, for any $n=1,\ldots,n^*$, the distribution of $Y_{m_n,j^*}$ depends only on the $n$th column of $\bm{Y}_{i^*,j^*}^{diff}$. Since $\begin{pmatrix} Y_{m_n,1}-Y_{\cdot,1}, & \ldots, & Y_{m_n,j^*}-Y_{\cdot j^*} \end{pmatrix}$ is symmetrically distributed about $\begin{pmatrix}0, & \ldots, & 0 \end{pmatrix}$, the distribution of $Y_{m_n,j^*} - Y_{\cdot j^*}$ conditional on $\bm{Y}_{m_n,j^*} - \bm{Y}_{\cdot j^*} = (a_1,\ldots,a_{j^*})$ for any constants $a_j$ is equal to the distribution of $-(Y_{m_n,j^*} - Y_{\cdot j^*})$ conditional on $\bm{Y}_{m_n,j^*} - \bm{Y}_{\cdot j^*} = (-a_1,\ldots,-a_{j^*})$. Hence:
\begin{equation}
    E[Y_{m_n,j^*} - Y_{\cdot j}|\bm{Y}_{i^*,j^*}^{diff} = \bm{A}] = -E[Y_{m_n,j^*} - Y_{\cdot j}|\bm{Y}_{i^*,j^*}^{diff} = -\bm{A}], \label{ExpFlip}
\end{equation}
for any difference matrix $\bm{A}$.

For any difference matrix $\bm{A}$ and any donor cluster $m_n$, then:
\begin{align}
    E[Y_{m_n,j^*}|&\bm{Y}_{i^*,j^*}^{diff} \in \{-\bm{A},\bm{A}\}] = P[\bm{Y}_{i^*,j^*}^{diff}=\bm{A}|\bm{Y}_{i^*,j^*}^{diff} \in \{-\bm{A},\bm{A}\}] E[Y_{m_n,j^*}|\bm{Y}_{i^*,j^*}^{diff} = \bm{A}] \nonumber \\
    &\qquad \qquad + P[\bm{Y}_{i^*,j^*}^{diff}=-\bm{A}|\bm{Y}_{i^*,j^*}^{diff} \in \{-\bm{A},\bm{A}\}] E[Y_{m_n,j^*}|\bm{Y}_{i^*,j^*}^{diff} = -\bm{A}] \nonumber \\
    &= \frac{1}{2} E[Y_{m_n,j^*}|\bm{Y}_{i^*,j^*}^{diff} = \bm{A}] + \frac{1}{2} E[Y_{m_n,j^*}|\bm{Y}_{i^*,j^*}^{diff} = -\bm{A}] \mbox{, by symmetry of } \bm{Y}_{i^*,j^*}^{diff} \nonumber \\
    &= Y_{\cdot j^*} + \frac{1}{2} \left( E[Y_{m_n,j^*} - Y_{\cdot j^*}|\bm{Y}_{i^*,j^*}^{diff} = \bm{A}] + E[Y_{m_n,j^*} - Y_{\cdot j^*}|\bm{Y}_{i^*,j^*}^{diff} = -\bm{A}] \right) \nonumber \\
    &= Y_{\cdot j^*} + \frac{1}{2} (0) \mbox{, by equation (\ref{ExpFlip})} \nonumber \\
    &= Y_{\cdot j^*}. \label{ExpnY0j*}
\end{align}
By equation (\ref{vijDefn}) and since $\norm*{\bm{Y}_{i^*,j^*}^{diff} \bm{v}} = \norm*{-\bm{Y}_{i^*,j^*}^{diff} \bm{v}}$ for all $\bm{v} \in \mathcal{V}$, $\bm{v}_{i^*,j^*}$ is correlated with the outcome vector $(Y_{m_1,j^*},~Y_{m_2,j^*},~\ldots,~Y_{m_{n^*},j^*})$ only through the elementwise absolute value of $\bm{Y}_{i^*,j^*}^{diff}$. Hence, for any $n=1,\ldots,n^*$:
\begin{align}
    E[Y_{m_n,j^*} | \bm{v}_{i^*,j^*}] &= E \left[ \left. E \left[ \left. Y_{m_n,j^*} \right| \bm{v}_{i^*,j^*} , \bm{Y}_{i^*,j^*}^{diff} \in \{\bm{A},-\bm{A}\} \right] \right|\bm{v}_{i^*,j^*} \right] \nonumber \\
    &= E \left[ \left. E \left[ \left. Y_{m_n,j^*} \right| \bm{Y}_{i^*,j^*}^{diff} \in \{\bm{A},-\bm{A}\} \right] \right| \bm{v}_{i^*,j^*} \right] \nonumber \\
    &= E \left[ \left. Y_{\cdot j^*} \right| \bm{v}_{i^*,j^*} \right] \mbox{, by equation (\ref{ExpnY0j*})} \nonumber \\
    &= Y_{\cdot j^*} \mbox{, since } Y_{\cdot j^*} \mbox{ is fixed.} \label{Exp1}
\end{align}

And thus, denoting by $v_{i^*,j^*,n}$ the $n$th component of the vector $\bm{v}_{i^*,j^*}$:
\begin{align}
E[Z_{i^*,j^*}] &=  E \left[ \sum_{n=1}^{n^*} v_{i^*,j^*,n} Y_{m_n, j^*} \right] = E \left[ E \left[ \left. \sum_{n=1}^{n^*} v_{i^*,j^*,n} Y_{m_n, j^*} \right| \bm{v}_{i^*,j^*} \right] \right] \nonumber \\
&= \sum_{n=1}^{n^*} E \left[ v_{i^*,j^*,n} E \left[ Y_{m_n,j^*} | \bm{v}_{i^*,j^*} \right] \right] \nonumber = \sum_{n=1}^{n^*} E \left[ v_{i^*,j^*,n} Y_{\cdot j^*} \right] \mbox{, by equation (\ref{Exp1})} \nonumber \\
&= Y_{\cdot j^*} E \left[ \sum_{n=1}^{n^*} v_{i^*,j^*,n} \right] = Y_{\cdot j^*} \cdot 1 = Y_{\cdot j^*} \mbox{, as desired}.
\end{align}
\end{proof}

\begin{corollary} \label{SCCor1}
Suppose that for each control cluster-period $(i,j)$, the individual outcomes $Y_{i,j,k}$ are independent and identically distributed, conditional on the cluster and period, with expectation $Y'_{i,j}$ and finite variance. Suppose further that the $Y'_{i,j}$ values satisfy the conditions in Theorem \ref{SCThm1}; that is, they are symmetrically distributed about some common expectation vector $\bm{Y}_{\cdot J}$ and each cluster's values are independent of the values from all other clusters. Then the synthetic control estimator $Z_{i^*,j^*}$ for any cluster $i^*$ in any period $j^*$ where that cluster is on intervention is an asymptotically (with respect to the number of individuals in each cluster) unbiased estimate of $Y_{\cdot j^*}$.
\end{corollary}

\begin{proof}
If the individual outcomes are independent and identically distributed with expectation $Y'_{i,j}$ and finite variance, then, by the Central Limit Theorem, the mean outcome $Y_{i,j}$ for each control cluster-period is asymptotically (with respect to the number of individuals in each cluster) normally distributed with expectation $Y'_{i,j}$ and finite variance. Thus, for any cluster $i$, we can write the distribution of the vector of pre-intervention cluster-level outcomes, as:
\begin{equation}
    (Y_{i,1},~\ldots,~Y_{i,j_i}) \overset{\mathcal{D}}{=} (Y'_{i,1},~\ldots,~Y'_{i,j_i}) + (B_{i,1},~\ldots,~B_{i,j_i}) + o_p(1),
\end{equation}
where $B_{i,j} \sim N(0,\xi_{i,j}^2)$ for some finite $\xi_{i,j}$, $o_p(1) \overset{P}{\rightarrow} \bm{0}$, and the $B_{i,j}$ are mutually independent. Since the $B_{i,j}$ are normally (and hence symmetrically) distributed, the limiting distribution of $(Y_{i,1},\ldots,Y_{i,j_i})$ is symmetric about $(Y_{\cdot 1},\ldots,Y_{\cdot j_i})$ by the assumption on $Y'_{i,j}$. Moreover, since the individual outcomes are independent conditional on the cluster-period mean and the cluster means are independent by assumption, $(Y_{i,1},\ldots,Y_{i,j_i}) \indep (Y_{i',1},\ldots,Y_{i',j_{i'}})$ for any $i \neq i'$.

Because of this asymptotic symmetry, for any target cluster-period ($i^*,j^*$) where $X_{i^*,j^*} = 1$, for any difference matrix $\bm{A}$:
\begin{align}
    &\lim_{K \to \infty} P[\bm{Y}_{i^*,j^*}^{diff} = \bm{A}|\bm{Y}_{i^*,j^*}^{diff} \in \{\bm{A},-\bm{A}\}] = \lim_{K \to \infty} P[\bm{Y}_{i^*,j^*}^{diff} = -\bm{A}|\bm{Y}_{i^*,j^*}^{diff} \in \{\bm{A},-\bm{A}\}] \mbox{, and so} \nonumber \\
    &P[\bm{Y}_{i^*,j^*}^{diff} = \bm{A}|\bm{Y}_{i^*,j^*}^{diff} \in \{\bm{A},-\bm{A}\}] = 1/2 + o(1), \label{Asyf}
\end{align}
where $\lim_{K \to \infty} o(1) = 0$. Additionally, by this symmetry, for any donor cluster $m_n \in \{m_1,\ldots,m_{n^*}\}$ (defined as in Theorem \ref{SCThm1}):
\begin{equation}
    E \left[ \left. Y_{m_n,j^*} - Y_{\cdot j^*} \right| \bm{Y}_{i^*,j^*}^{diff} = \bm{A} \right] = - E \left[ \left. Y_{m_n,j^*} - Y_{\cdot j^*} \right| \bm{Y}_{i^*,j^*}^{diff} = -\bm{A} \right] + o(1). \label{AsyE}
\end{equation}
Thus, for any difference matrix $\bm{A}$:
\begin{align*}
    &E \left[ \left. Y_{m_n,j^*} \right| \bm{Y}_{i^*,j^*}^{diff} \in \{\bm{A},-\bm{A}\} \right] = P[\bm{Y}_{i^*,j^*}^{diff} = \bm{A}|\bm{Y}_{i^*,j^*}^{diff} \in \{\bm{A},-\bm{A}\}] E[Y_{m_n, j^*}|Y_{i^*, j^*}^{diff} = A] \\
    &\qquad \qquad + P[\bm{Y}_{i^*,j^*}^{diff} = -\bm{A}|\bm{Y}_{i^*,j^*}^{diff} \in \{\bm{A},-\bm{A}\}] E[Y_{m_n, j^*}|\bm{Y}_{i^*, j^*}^{diff} = -\bm{A}] \\
    &= \left(\frac{1}{2} + o(1) \right) \left( E[Y_{m_n, j^*}|\bm{Y}_{i^*, j^*}^{diff} = \bm{A}] + E[Y_{m_n, j^*}|\bm{Y}_{i^*, j^*}^{diff} = -\bm{A}] \right) \mbox{, by equation (\ref{Asyf})} \\
    &= \left(\frac{1}{2} + o(1) \right) \left( 2Y_{\cdot j^*} + E[Y_{m_n, j^*} - Y_{\cdot j^*}|\bm{Y}_{i^*, j^*}^{diff} = \bm{A}] + E[Y_{m_n, j^*} - Y_{\cdot j^*}|\bm{Y}_{i^*, j^*}^{diff} = -\bm{A}] \right) \\
    &= \left(\frac{1}{2} + o(1) \right) \left( 2 Y_{\cdot j^*} + o(1) \right) \mbox{, by equation (\ref{AsyE})} \\
    &= Y_{\cdot j^*} + o(1) \mbox{, by the properties of convergence.}
\end{align*}
And so for any difference matrix $\bm{A}$ for any $n=1,\ldots,n^*$:
\begin{equation}
    E \left[ \left. Y_{m_n,j^*} \right| \bm{Y}_{i^*,j^*}^{diff} \in \{\bm{A},-\bm{A}\} \right] = Y_{\cdot j^*} + o(1).
\end{equation}
Following the steps in the proof of Theorem \ref{SCThm1}, and using the properties of convergence, then, for any $n=1,\ldots,n^*$:
\begin{equation}
    E \left[ \left. Y_{m_n,j^*} \right| \bm{v}_{i^*,j^*} \right] = Y_{\cdot j^*} + o(1),
\end{equation}
and thus, using that $\lim_{K \to \infty} o(1) = 0$:
\begin{equation}
    \lim_{K \to \infty} E[Z_{i^*,j^*}] = \lim_{K \to \infty} E \left[ \sum_{n=1}^{n^*} v_{i^*,j^*,n} Y_{m_n,j^*} \right] = \lim_{K \to \infty} \left(Y_{\cdot j^*} + o(1) \right) = Y_{\cdot j^*}.
\end{equation}
Hence, for any cluster $i^*$ and period $j^*$ where $X_{i^*,j^*} = 1$, $Z_{i,j}$ is an asymptotically unbiased estimate of $Y_{\cdot j^*}$.
\end{proof}

\begin{theorem} \label{SCThm2}
Assume that the assumptions of Theorem 1 are met and that for any intervention cluster-period $(i,j)$, $E[Y_{i,j}] = Y_{\cdot j} + \beta$. Then for the risk difference function, $g(y_1,y_2) = y_1 - y_2$, the synthetic control estimator $\hat{\beta}$ with weights $w_{i,j}$ independent of the outcomes is an unbiased estimate of $\beta$.
\end{theorem}

\begin{proof}
By Theorem \ref{SCThm1}, for any target cluster-period ($i,j$) such that $X_{i,j} = 1$, $E[Z_{i,j}] = Y_{\cdot j}$ (note that we have dropped the $i^*,j^*$ notation for simplicity). Thus:
\begin{align}
    E \left[\hat{\beta} \right] &= E \left[ \sum_{(i,j): X_{i,j} = 1} \frac{w_{i,j}}{\sum_{(i,j): X_{i,j} = 1} w_{i,j}} \hat{\beta}_{i,j} \right] = \sum_{(i,j): X_{i,j} = 1} E \left[ \frac{w_{i,j}}{\sum_{(i,j): X_{i,j} = 1} w_{i,j}} (Y_{i,j} - Z_{i,j}) \right] \nonumber \\
    &= \sum_{(i,j): X_{i,j} = 1} \frac{w_{i,j}}{\sum_{(i,j): X_{i,j} = 1} w_{i,j}} E[Y_{i,j} - Z_{i,j}] \mbox{, since } w_{i,j} \indep Y_{i,j},Z_{i,j} \nonumber \\
    &= \sum_{(i,j): X_{i,j} = 1} \frac{w_{i,j}}{\sum_{(i,j): X_{i,j} = 1} w_{i,j}} \left( Y_{\cdot j} + \beta - Y_{\cdot j} \right) = \frac{\sum_{(i,j): X_{i,j} = 1} w_{i,j}}{\sum_{(i,j): X_{i,j} = 1} w_{i,j}} \beta = \beta.
\end{align}
\end{proof}

\begin{corollary} \label{SCCor2}
Assume that the assumptions of Corollary 1 are met and that for any intervention cluster-period $(i,j)$, $E[Y_{i,j}] = Y_{\cdot j} + \beta$. Then for the risk difference function, $g(y_1,y_2) = y_1 - y_2$, the synthetic control estimator $\hat{\beta}$ with weights $w_{i,j}$ independent of the outcomes is an asymptotically (with respect to the number of individuals in each cluster) unbiased estimate of $\beta$.
\end{corollary}

\begin{proof}
By Corollary \ref{SCCor1}, for any target cluster-period ($i,j$) such that $X_{i,j} = 1$, $\lim_{K \to \infty} E[Z_{i,j}] = Y_{\cdot j}$ (again dropping the $i^*,j^*$ notation). Thus:
\begin{align}
    \lim_{K \to \infty} E \left[ \hat{\beta} \right] &= \lim_{K \to \infty} E \left[ \sum_{(i,j): X_{i,j} = 1} \frac{w_{i,j}}{\sum_{(i,j): X_{i,j} = 1} w_{i,j}} \hat{\beta}_{i,j} \right] \nonumber\\
    &= \sum_{(i,j): X_{i,j} = 1} \lim_{K \to \infty} E \left[ \frac{w_{i,j}}{\sum_{(i,j): X_{i,j} = 1} w_{i,j}} (Y_{i,j} - Z_{i,j}) \right] \nonumber \\
    &= \sum_{(i,j): X_{i,j} = 1} \frac{w_{i,j}}{\sum_{(i,j): X_{i,j} = 1} w_{i,j}} \lim_{K \to \infty} E[Y_{i,j} - Z_{i,j}] \mbox{, since } w_{i,j} \indep Y_{i,j},Z_{i,j},K \nonumber \\
    &= \sum_{(i,j): X_{i,j} = 1} \frac{w_{i,j}}{\sum_{(i,j): X_{i,j} = 1} w_{i,j}} \left( \lim_{K \to \infty} E[Y_{i,j}] - Y_{\cdot j} \right) \nonumber \\
    &= \sum_{(i,j): X_{i,j} = 1} \frac{w_{i,j}}{\sum_{(i,j): X_{i,j} = 1} w_{i,j}} \left( Y_{\cdot j} + \beta - Y_{\cdot j} \right) \nonumber\\
    &= \frac{\sum_{(i,j): X_{i,j} = 1} w_{i,j}}{\sum_{(i,j): X_{i,j} = 1} w_{i,j}} \beta = \beta.
\end{align}
Thus $\hat{\beta}$ is an asymptotically unbiased estimate of $\beta$.
\end{proof}

\begin{theorem} \label{COThm1}
Assume that there is a constant risk difference $\beta$ due to treatment across clusters and periods; that is, $E[Y_{i,j}|X_{i,j} = 1] = E[Y_{i,j}|X_{i,j} = 0] + \beta$ for all $i$,$j$. Then for any weights $w_j$ that are independent of the outcomes $Y_{i,j}$, the crossover estimators $\hat{\beta}$ and $\tilde{\beta}$ using the risk difference function, $g(y_1,y_2) = y_1 - y_2$, are unbiased estimates of $\beta$. That is, $E[\hat{\beta}] = E[\tilde{\beta}] = \beta$.
\end{theorem}

\begin{proof}
We denote by $Y_{\cdot j}$ the expectation (marginal across clusters) of the outcome of any cluster on control in period $j$. By the assumptions, for all $j > 1$:
\begin{align}
    E[D_{i,j}|i \in I_{0,j}] &= E[Y_{i,j}|X_{i,j}=0] - E[Y_{i,j-1}|X_{i,j-1} = 0] = Y_{\cdot j} - Y_{\cdot j-1}. \\
    E[D_{i,j}|i \in I_{1,j}] &= E[Y_{i,j}|X_{i,j}=1] - E[Y_{i,j-1}|X_{i,j-1}=0] \nonumber \\
    &= E[Y_{i,j}|X_{i,j}=1] - E[Y_{i,j}|X_{i,j}=0] + E[Y_{i,j}|X_{i,j}=0] - E[Y_{i,j-1}|X_{i,j-1}=0] \nonumber \\
    &= \beta + Y_{\cdot j} - Y_{\cdot j-1}. \\
    E[D_{i,j}|i \in I_{2,j}] &= E[Y_{i,j}|X_{i,j}=1] - E[Y_{i,j-1}|X_{i,j-1}=1] \nonumber\\
    &= E[Y_{i,j}|X_{i,j}=1] - E[Y_{i,j}|X_{i,j}=0] + E[Y_{i,j}|X_{i,j}=0] \nonumber \\
    &\quad - E[Y_{i,j-1}|X_{i,j-1}=0]+ E[Y_{i,j-1}|X_{i,j-1}=0] - E[Y_{i,j-1}|X_{i,j-1}=1] \nonumber \\
    &= \beta + Y_{\cdot j} - Y_{\cdot j-1} - \beta = Y_{\cdot j} - Y_{\cdot j-1}.
\end{align}

Define $\theta'_j = Y_{\cdot j} - Y_{\cdot j-1}$ for all $j > 1$. Then:
\begin{align}
    E[\hat{\beta}_j] &= \sum_{i \in I_{1,j}} \frac{1}{n_{1,j}} E[D_{i,j}|i \in I_{1,j}] - \sum_{i \in I_{0,j}} \frac{1}{n_{0,j}} E[D_{i,j}|i \in I_{0,j}] \nonumber \\
    &= \frac{n_{1,j}}{n_{1,j}} (\beta + \theta'_j) - \frac{n_{0,j}}{n_{0,j}} (\theta'_j) = \beta. \label{ExpCO1}\\
    E[\tilde{\beta}_j] &= \sum_{i \in I_{1,j}} \frac{1}{n_{1,j}} E[D_{i,j}|i \in I_{1,j}] - \sum_{i \in I_{0,j}} \frac{1}{n_{0,j}+n_{2,j}} E[D_{i,j}|i \in I_{0,j}] \nonumber \\
    &\qquad - \sum_{i \in I_{2,j}} \frac{1}{n_{0,j}+n_{2,j}} E[D_{i,j}|i \in I_{2,j}] \nonumber \\
    &= \frac{n_{1,j}}{n_{1,j}} (\beta + \theta'_j)  - \frac{n_{0,j}}{n_{0,j}+n_{2,j}} (\theta'_j) - \frac{n_{2,j}}{n_{0,j}+n_{2,j}} (\theta'_j) = \beta. \label{ExpCO2}
\end{align}

Now, for any weights $w_{j}$ that are independent of the outcomes:
\begin{align}
    E[\hat{\beta}] &= E \left[ \sum_{j > 1} \frac{w_j}{w} \hat{\beta}_j \right] = \sum_{j > 1} \frac{w_j}{w} \beta \mbox{, by equation (\ref{ExpCO1})} \nonumber \\
    &= \frac{\sum_{j > 1} w_j}{w} \beta = \beta. \label{ExpCO3} \\
    E[\tilde{\beta}] &= E \left[ \sum_{j > 1} \frac{w_j}{w} \tilde{\beta}_j \right] = \sum_{j > 1} \frac{w_j}{w} \beta \mbox{, by equation (\ref{ExpCO2})} \nonumber \\
    &= \frac{\sum_{j > 1} w_j}{w} \beta = \beta. \label{ExpCO4}
\end{align}

So as desired, $E[\hat{\beta}] = E[\tilde{\beta}] = \beta$.
\end{proof}

\begin{corollary} \label{COCor1}
Assume that there is a constant risk difference $\beta$ due to treatment in the first period on treatment across clusters; that is $E[Y_{i,j}|X_{i,j} = 1 ~ \cap ~ X_{i,j-1} = 0] = E[Y_{i,j}|X_{i,j} = 0] + \beta$ for all $i,j$. Then for any weights $w_j$ that are independent of the outcomes $Y_{i,j}$, the crossover estimator $\hat{\beta}$ using the risk difference function, $g(y_1,y_2) = y_1 - y_2$, is an unbiased estimate of $\beta$. That is, $E[\hat{\beta}] = \beta$.
\end{corollary}

\begin{proof}
Again, we denote by $Y_{\cdot j}$ the expectation (marginal across clusters) of the outcome of any cluster on control in period $j$. By the assumptions, for all $j > 1$:
\begin{align}
    E[D_{i,j}|i \in I_{0,j}] &= E[Y_{i,j}|X_{i,j}=0] - E[Y_{i,j-1}|X_{i,j-1} = 0] = Y_{\cdot j} - Y_{\cdot j-1}. \\
    E[D_{i,j}|i \in I_{1,j}] &= E[Y_{i,j}|X_{i,j}=1] - E[Y_{i,j-1}|X_{i,j-1}=0] \nonumber \\
    &= E[Y_{i,j}|X_{i,j}=1] - E[Y_{i,j}|X_{i,j}=0] + E[Y_{i,j}|X_{i,j}=0] \nonumber \\
    &\qquad - E[Y_{i,j-1}|X_{i,j-1}=0] \nonumber \\
    &= \beta + Y_{\cdot j} - Y_{\cdot j-1}.
\end{align}

Define $\theta'_j = Y_{\cdot j} - Y_{\cdot j-1}$ for all $j > 1$. Then:
\begin{align}
    E[\hat{\beta}_j] &= \sum_{i \in I_{1,j}} \frac{1}{n_{1,j}} E[D_{i,j}|i \in I_{1,j}] - \sum_{i \in I_{0,j}} \frac{1}{n_{0,j}} E[D_{i,j}|i \in I_{0,j}] \nonumber \\
    &= \frac{n_{1,j}}{n_{1,j}} (\beta + \theta'_j) - \frac{n_{0,j}}{n_{0,j}} (\theta'_j) = \beta. \label{ExpCOCor1}
\end{align}

Now, for any weights $w_{j}$ that are independent of the outcomes, by equation (\ref{ExpCOCor1}):
\begin{align}
    E[\hat{\beta}] &= E \left[ \sum_{j > 1} \frac{w_j}{w} \hat{\beta}_j \right] = \sum_{j > 1} \frac{w_j}{w} \beta = \frac{\sum_{j > 1} w_j}{w} \beta = \beta. \label{ExpCOCor3}
\end{align}

So as desired, $E[\hat{\beta}] = \beta$.

\end{proof}

\begin{remark}
Since $E[\hat{\beta}]$ depends only on $E[D_{i,j}|i \in I_{1,j}]$ and $E[D_{i,j}|i \in I_{0,j}]$, it requires only the weaker assumption of Corollary \ref{COCor1} to be unbiased, while $E[\tilde{\beta}]$ requires the stronger assumption given in Theorem \ref{COThm1}.
\end{remark}

\begin{remark}
Specifically, equal weighting and the weights $w_j$ and $w'_j$ given in Section \ref{COweights} are independent of the outcomes $Y_{i,j}$ and thus result in unbiased estimates if the other conditions of Theorem \ref{COThm1} are met.
\end{remark}

\clearpage


\section{Variance and covariance of estimators} \label{app2}
\setcounter{equation}{0}
\renewcommand{\theequation}{B\arabic{equation}}
\renewcommand{\thesection}{B}

\subsection{Simplified data-generating setting}
Throughout this section, we assume that the underlying data-generating process for the cluster-level outcomes follows a mixed-effects model:
\begin{equation}
    Y_{i,j} = \mu + \alpha_i + \theta_j + X_{i,j} \beta + \epsilon_{i,j}
\end{equation}
where the $\alpha_i$ are independently distributed with mean 0 and variance $\tau^2$, the $\epsilon_{i,j}$ are independently distributed with mean 0 and variance $\sigma^2_{i,j}$ and $\alpha_i \indep \epsilon_{i',j'}$ for all $i,i',j'$. Note that this is similar to the risk difference model in Simulation 1 in Section \ref{Results}, but differs in that we force $\epsilon_{i,j} \indep \alpha_i$, so the mean-variance relationship no longer holds.

Thus:
\begin{align*}
    \V(Y_{i,j}) &= \tau^2 + \sigma^2_{i,j}. \\
    \Cov(Y_{i,j},~Y_{i',j'}) &= \begin{cases} \tau^2,& ~i = i', j \neq j' \\ 0,& ~i \neq i' \end{cases}.
\end{align*}

We assume, as in the simulations in Section \ref{Results}, that there are $J$ periods, $J-1$ clusters, with cluster $i$ on control for periods $1,\ldots,i$ and on intervention for periods $i+1,\ldots,J$.

\subsection{Variance of the non-parametric within-period estimator}

The non-parametric within-period estimator is given by:
\begin{equation}
    \hat{\beta}^{NPWP} = \sum_{j=2}^{J-1} v_j \left( \frac{\sum_{i=1}^{j-1} Y_{i,j}}{j-1} - \frac{\sum_{i=j}^{J-1} Y_{i,j}}{J-j} \right) \equiv \sum_{j=2}^{J-1} v_j A_j,
\end{equation}
for weights $v_j$ that sum to 1. For ease of variance calculations, we here assume, contrary to the estimator used Section \ref{Results}, that the weights $v_j$ are fixed and independent of the data (they can still depend on the number of clusters in the intervention and control conditions in period $j$, however).

\begin{align*}
    \V(A_j) &= \V \left( \frac{\sum_{i=1}^{j-1} Y_{i,j}}{j-1} - \frac{\sum_{i=j}^{J-1} Y_{i,j}}{J-j} \right) = \sum_{i=1}^{j-1} \frac{\V(Y_{i,j})}{(j-1)^2} + \sum_{i=j}^{J-1} \frac{\V(Y_{i,j})}{(J-j)^2} \\
    &= \frac{\sum_{i=1}^{j-1} (\sigma^2_{i,j} + \tau^2)}{(j-1)^2} + \frac{\sum_{i=j}^{J-1} (\sigma^2_{i,j} + \tau^2)}{(J-j)^2} = \frac{J-1}{(J-j) (j-1)} \tau^2 + \frac{\sum_{i=1}^{j-1} \sigma^2_{i,j}}{(j-1)^2} + \frac{\sum_{i=j}^{J-1} \sigma^2_{i,j}}{(J-j)^2}.
\end{align*}
For $j \neq k$ (WLOG, assume $k > j$):
\begin{align*}
    \Cov(A_j,A_k) &= \Cov \left( \frac{\sum_{i=1}^{j-1} Y_{i,j}}{j-1} - \frac{\sum_{i=j}^{J-1} Y_{i,j}}{J-j}, ~\frac{\sum_{i=1}^{k-1} Y_{i,k}}{k-1} - \frac{\sum_{i=k}^{J-1} Y_{i,k}}{J-k} \right) \\
    &= \frac{1}{(j-1) (k-1)} \sum_{i=1}^{j-1} \Cov(Y_{i,j},~Y_{i,k}) - \frac{1}{(J-j) (k-1)} \sum_{i=j}^{k-1} \Cov(Y_{i,j},~Y_{i,k}) \\
    &\qquad + \frac{1}{(J-j) (J-k)} \sum_{i=k}^{J-1} \Cov(Y_{i,j},~Y_{i,k}) \\
    &= \tau^2 \left( \frac{1}{k-1} - \frac{k-j}{(J-j) (k-1)} + \frac{1}{J-j} \right) \\
    &= \frac{\tau^2}{(J-j) (k-1)} \left(J-j - (k-j) + k-1 \right) = \tau^2 \frac{J-1}{(J-j) (k-1)}.
\end{align*}
Thus:
\begin{align}
    \V(\hat{\beta}^{NPWP}) &= \Cov \left( \sum_{j=2}^{J-1} v_j A_j,~ \sum_{k=2}^{J-1} v_k A_k \right) \nonumber \\
    &= \sum_{j=2}^{J-1} v_j^2 \V(A_j) + 2 \sum_{j=2}^{J-2} \sum_{k=j+1}^{J-1} v_j v_k \Cov(A_j,~A_k) \nonumber \\
    &= \sum_{j=2}^{J-1} v_j^2 \left[ \frac{J-1}{(J-j) (j-1)} \tau^2 + \frac{\sum_{i=1}^{j-1} \sigma^2_{i,j}}{(j-1)^2} + \frac{\sum_{i=j}^{J-1} \sigma^2_{i,j}}{(J-j)^2} \right] \nonumber \\
    &\quad \qquad + 2 \sum_{j=2}^{J-2} \sum_{k=j+1}^{J-1} v_j v_k \tau^2 \frac{J-1}{(J-j) (k-1)}.
\end{align}

Note that the variance increases as any $\sigma^2_{i,j}$ or $\tau^2$ increase, holding all else constant.

\subsection{Variance of the crossover estimator}
The crossover estimator is given by:
\begin{equation}
    \hat{\beta}^{CO} = \sum_{j=2}^{J-1} w_j \left[ Y_{j-1,j} - Y_{j-1,j-1} - \frac{1}{J-j} \sum_{\ell=j}^{J-1} (Y_{\ell,j} - Y_{\ell,j-1}) \right] \equiv \sum_{j=2}^{J-1} w_j B_j,
\end{equation}
for weights $w_j$ that sum to 1.

Note that for all $i$, $Y_{i,j} - Y_{i,j'} = \epsilon_{ij} - \epsilon_{ij'} + \beta I(X_{ij} \neq X_{ij'})$. Hence:
\begin{align*}
    \V(B_j) &= \V \left( \epsilon_{j-1,j} - \epsilon_{j-1,j-1} + \beta - \frac{1}{J-j} \sum_{\ell=j}^{J-1} (\epsilon_{\ell,j} - \epsilon_{\ell,j-1}) \right) \\
    &= \sigma^2_{j-1,j} + \sigma^2_{j-1,j-1} + \frac{1}{(J-j)^2} \sum_{\ell=j}^{J-1} (\sigma^2_{\ell,j} + \sigma^2_{\ell,j-1}). \\
    \Cov(B_j,B_{j+1}) &= \Cov \left( \epsilon_{j-1,j} - \epsilon_{j-1,j-1} + \beta - \frac{1}{J-j} \sum_{\ell=j}^{J-1} (\epsilon_{\ell,j} - \epsilon_{\ell,j-1}), \right.\\
    &\qquad \qquad\left. \epsilon_{j,j+1} - \epsilon_{j,j} + \beta - \frac{1}{J-j-1} \sum_{m=j+1}^{J-1} (\epsilon_{m,j+1} - \epsilon_{m,j}) \right) \\
    &= \frac{1}{J-j} \left[\sigma^2_{j,j} - \sum_{\ell=j+1}^{J-1} \frac{\sigma^2_{\ell,j}}{J-j-1} \right].
\end{align*}
\begin{align*}
    \Cov(B_j,B_k) &= \Cov \left( \epsilon_{j-1,j} - \epsilon_{j-1,j-1} + \beta - \frac{1}{J-j} \sum_{\ell=j}^{J-1} (\epsilon_{\ell,j} - \epsilon_{\ell,j-1}), \right.\\
    &\qquad \qquad \left. \epsilon_{k-1,k} - \epsilon_{k-1,k-1} + \beta - \frac{1}{J-k} \sum_{m=k}^{J-1} (\epsilon_{m,k} - \epsilon_{m,k-1}) \right) \\
    &= 0 \mbox{ for } k > j+1.
\end{align*}
Thus:
\begin{align}
    \V(\hat{\beta}^{CO}) &= \Cov \left( \sum_{j=2}^{J-1} w_j B_j,~\sum_{k=2}^{J-1} w_k B_k \right) \nonumber\\
    &= \sum_{j=2}^{J-1} w_j^2 \V(B_j) + 2 \sum_{j=2}^{J-2} w_j w_{j+1} \Cov(B_j,~B_{j+1}) \nonumber\\
    &= \sum_{j=2}^{J-1} w_j^2 \left[ \sigma^2_{j-1,j} + \sigma^2_{j-1,j-1} + \frac{1}{(J-j)^2} \sum_{\ell=j}^{J-1} (\sigma^2_{\ell,j} + \sigma^2_{\ell,j-1}) \right] \nonumber\\
    &\qquad + 2 \sum_{j=2}^{J-2} w_j w_{j+1} \frac{1}{J-j} \left( \sigma^2_{j,j} -  \sum_{\ell=j+1}^{J-1} \frac{\sigma^2_{\ell,j}}{J-j-1} \right).
\end{align}

Note that, unlike the variance of the NPWP estimator, this variance does not depend on $\tau^2$.

\subsection{Covariance of the non-parametric within-period estimator and the crossover estimator}
We begin by noting that:
\begin{align*}
    A_j &= \frac{1}{j-1} \sum_{i=1}^{j-1} (\mu + \alpha_i + \theta_j + \beta + \epsilon_{i,j}) - \frac{1}{J-j} \sum_{i=j}^{J-1} (\mu + \alpha_i + \theta_j + \epsilon_{i,j}) \\
    &= \mu + \theta_j + \beta + \frac{1}{j-1} \sum_{i=1}^{j-1} (\alpha_i + \epsilon_{i,j}) - (\mu + \theta_j) - \frac{1}{J-j} \sum_{i=j}^{J-1} (\alpha_i + \epsilon_{i,j}) \\
    &= \beta + \frac{1}{j-1} \sum_{i=1}^{j-1} (\alpha_i + \epsilon_{i,j}) - \frac{1}{J-j} \sum_{i=j}^{J-1} (\alpha_i + \epsilon_{i,j}).
\end{align*}
Hence, dropping the $\beta$ terms in $A_j$ and $B_j$ since they do not contribute any variability:
\begin{align*}
    \Cov(A_j,~B_j) &= \frac{1}{j-1} \sum_{i=1}^{j-1} \Cov \left( \alpha_i + \epsilon_{i,j},~\epsilon_{j-1,j} - \epsilon_{j-1,j-1} - \frac{1}{J-j} \sum_{\ell=j}^{J-1} (\epsilon_{\ell,j}-\epsilon_{\ell,j-1}) \right) \\
    &\qquad -\frac{1}{J-j} \sum_{i=j}^{J-1} \Cov \left( \alpha_i + \epsilon_{i,j},~\epsilon_{j-1,j} - \epsilon_{j-1,j-1} - \frac{1}{J-j} \sum_{\ell=j}^{J-1} (\epsilon_{\ell,j} - \epsilon_{\ell,j-1}) \right) \\
    &= \frac{1}{j-1} \sigma^2_{j-1,j} + \frac{1}{(J-j)^2} \sum_{i=j}^{J-1} \sigma^2_{i,j}. \\
    \Cov(A_j,~B_k) &= \frac{1}{j-1} \sum_{i=1}^{j-1} \Cov \left( \alpha_i + \epsilon_{i,j},~\epsilon_{k-1,k} - \epsilon_{k-1,k-1} - \frac{1}{J-k} \sum_{\ell=k}^{J-1} (\epsilon_{\ell,k} - \epsilon_{\ell,k-1}) \right) \\
    &\qquad - \frac{1}{J-j} \sum_{i=j}^{J-1} \Cov \left(\alpha_i + \epsilon_{i,j},~\epsilon_{k-1,k} - \epsilon_{k-1,k-1} - \frac{1}{J-k} \sum_{\ell=k}^{J-1} (\epsilon_{\ell,k} - \epsilon_{\ell,k-1}) \right).
\end{align*}
For $j = k-1$:
\begin{align*}
    \Cov(A_j,~B_k) &= \frac{\sigma^2_{j,j}}{J-j} -  \frac{1}{(J-j)(J-j-1)} \sum_{\ell=j+1}^{J-1} \sigma^2_{\ell,j}.
\end{align*}
For $j < k-1$ or $j > k$, $\Cov(A_j,~B_k) = 0$.

Therefore:
\begin{align}
    \Cov \left(\hat{\beta}^{NPWP},~\hat{\beta}^{CO}\right) &= \Cov \left( \sum_{j=2}^{J-1} v_j A_j,~\sum_{k=2}^{J-1} w_k B_k \right) \nonumber \\
    &= \sum_{j=2}^{J-1} \sum_{k=2}^{J-1} v_j w_k \Cov(A_j,~B_k) \nonumber \\
    &= \sum_{j=2}^{J-1} v_j w_j \left[ \frac{1}{j-1} \sigma^2_{j-1,j} + \frac{1}{(J-j)^2} \sum_{i=j}^{J-1} \sigma^2_{i,j} \right] \nonumber \\
    &\qquad + \sum_{j=2}^{J-2} v_j w_{j+1} \left[ \frac{\sigma^2_{j,j}}{J-j} -  \frac{1}{(J-j)(J-j-1)} \sum_{i=j+1}^{J-1} \sigma^2_{i,j} \right].
\end{align}

Note that, like the variance of the crossover estimator, the covariance here does not depend on $\tau^2$.

\subsection{Conditions for the ensemble estimator to have lower variance}
If we moreover assume that $\sigma^2_{ij} = \sigma^2$ for all $i,j$, then:
\begin{align*}
    \V(\hat{\beta}^{NPWP}) &= (\tau^2 + \sigma^2) \sum_{j=2}^{J-1} v_j^2 \left[ \frac{J-1}{(J-j) (j-1)} \right] + 2 \tau^2 \sum_{j=2}^{J-2} \sum_{k=j+1}^{J-1} v_j v_k \frac{J-1}{(J-j) (k-1)}, \\
    &= (\tau^2 + \sigma^2) \sum_{j=2}^{J-1} v_j^2 \left( \frac{1}{j-1} + \frac{1}{J-j} \right) + 2 \tau^2 \sum_{j=2}^{J-2} \sum_{k=j+1}^{J-1} v_j v_k \frac{J-1}{(J-j) (k-1)}. \\
    \V(\hat{\beta}^{CO}) &= 2 \sum_{j=2}^{J-1} w_j^2 \sigma^2 \left(1 + \frac{1}{J-j} \right).\\
    \Cov(\hat{\beta}^{NPWP},~\hat{\beta}^{CO}) &= \sum_{j=2}^{J-1} v_j w_j \sigma^2 \left( \frac{1}{j-1} + \frac{1}{J-j} \right) < \sum_{j=2}^{J-1} v_j w_j \sigma^2 \left(1 + \frac{1}{J-j} \right).\\
    \Cov(\hat{\beta}^{NPWP},~\hat{\beta}^{CO}) &= \sum_{j=2}^{J-1} v_j w_j \sigma^2 \left( \frac{1}{j-1} + \frac{1}{J-j} \right) \\
    &< \sigma^2 \sum_{j=2}^{J-1} v_j w_j \left( \frac{1}{j-1} + \frac{1}{J-j} \right) + 2 \tau^2 \sum_{j=2}^{J-2} \sum_{k=j+1}^{J-1} v_j v_k \frac{J-1}{(J-j) (k-1)}.
\end{align*}

If $v_j \le 2 w_j$ for all $j=1,\ldots,n$, then $\Cov(\hat{\beta}^{NPWP},~\hat{\beta}^{CO}) < \V(\hat{\beta}^{CO})$. If $\sigma^2 w_j < (\tau^2 + \sigma^2) v_j$ for all $j=1,\ldots,n$, then $\Cov(\hat{\beta}^{NPWP},~\hat{\beta}^{CO}) < \V(\hat{\beta}^{NPWP})$.

Now consider the ensemble estimator $\hat{\beta}^{ENS} = \frac{1}{2} \hat{\beta}^{NPWP} + \frac{1}{2} \hat{\beta}^{CO}$. Then:
\[ \V(\hat{\beta}^{ENS}) = \frac{1}{4} \V(\hat{\beta}^{NPWP}) + \frac{1}{4} \V(\hat{\beta}^{CO}) + \frac{1}{2} \Cov(\hat{\beta}^{NPWP},~\hat{\beta}^{CO}). \]
So a necessary and sufficient condition for $\V(\hat{\beta}^{ENS}) < \min \left(\V(\hat{\beta}^{NPWP}),~\V(\hat{\beta}^{CO})\right)$ is:
\[ 2~\Cov(\hat{\beta}^{NPWP},~\hat{\beta}^{CO}) < 3 \min \left(\V(\hat{\beta}^{NPWP}),~\V(\hat{\beta}^{CO})\right) - \max \left(\V(\hat{\beta}^{NPWP}),~\V(\hat{\beta}^{CO})\right). \]

For $\tau^2 = 0$, $\V(\hat{\beta}^{NPWP} < \hat{\beta}^{CO})$ for some values of $v_j$, $w_j$, and $\sigma^2$ (e.g., for $v_j = w_j$ for all $\sigma^2$). Take such a set of values of $v_j$, $w_j$, and $\sigma^2$. Then, since $\V(\hat{\beta}^{NPWP})$ is an increasing function of $\tau^2$ and $\V(\hat{\beta}^{CO})$ is independent of $\tau^2$, there is a $\tau^2$ such that $\V(\hat{\beta}^{NPWP}) = \V(\hat{\beta}^{CO})$. So the condition simplifies to:
\[ 2~\Cov(\hat{\beta}^{NPWP},~\hat{\beta}^{CO}) < 2~\V(\hat{\beta}^{CO}). \]
For $v_j \le 2 w_j$ for all $j=2,\ldots,J-1$, $\Cov(\hat{\beta}^{NPWP},~\hat{\beta}^{CO}) < \V(\hat{\beta}^{CO})$. Hence, in this setting, the ensemble estimator has a lower variance than either single estimator.

This suggests that when $\tau^2$ is low enough that the crossover estimator and the non-parametric within-period estimator to have similar variances, and the $v_j$ and $w_j$ are close enough for the covariance to be lower than either of these variances, the ensemble method has lower variance than either single estimator.

\clearpage
\renewcommand{\thesection}{C}

\section{Goodness of fit of mixed effects models} \label{app3}
\setcounter{equation}{0}
\setcounter{figure}{0}
\renewcommand{\theequation}{C\arabic{equation}}
\renewcommand{\thefigure}{C\arabic{figure}}

Figure \ref{FigC1} provides two graphical displays of the goodness of the model fit: the quantile-quantile plot of the standardized estimated random effects compared to a normal distribution and the plot of residuals; both are shown for the risk difference (identity link) model, but are similar for the odds ratio (logit link) model.

These figures suggest violations of the model assumptions, including potential non-normality of the random effects. The Mahalanobis distance from the estimated random effects to the normal distribution,\cite{Singer2013-op} however, is not significant at the 0.05 level. Common goodness-of-fit tests applied to this distribution, such as the Kolmogorov-Smirnov, Anderson-Darling, and Cram\'er-von Mises tests,\cite{Ritz2004-yq} also yield non-significant results (all with $p$-values above $0.65$). The power of these tests to detect violations of the assumptions with a small number of clusters can be quite low, however.\cite{Yap2011-bu}

\begin{center}
\begin{figure}[!ht]
\centering
\includegraphics[width=.83\textwidth]{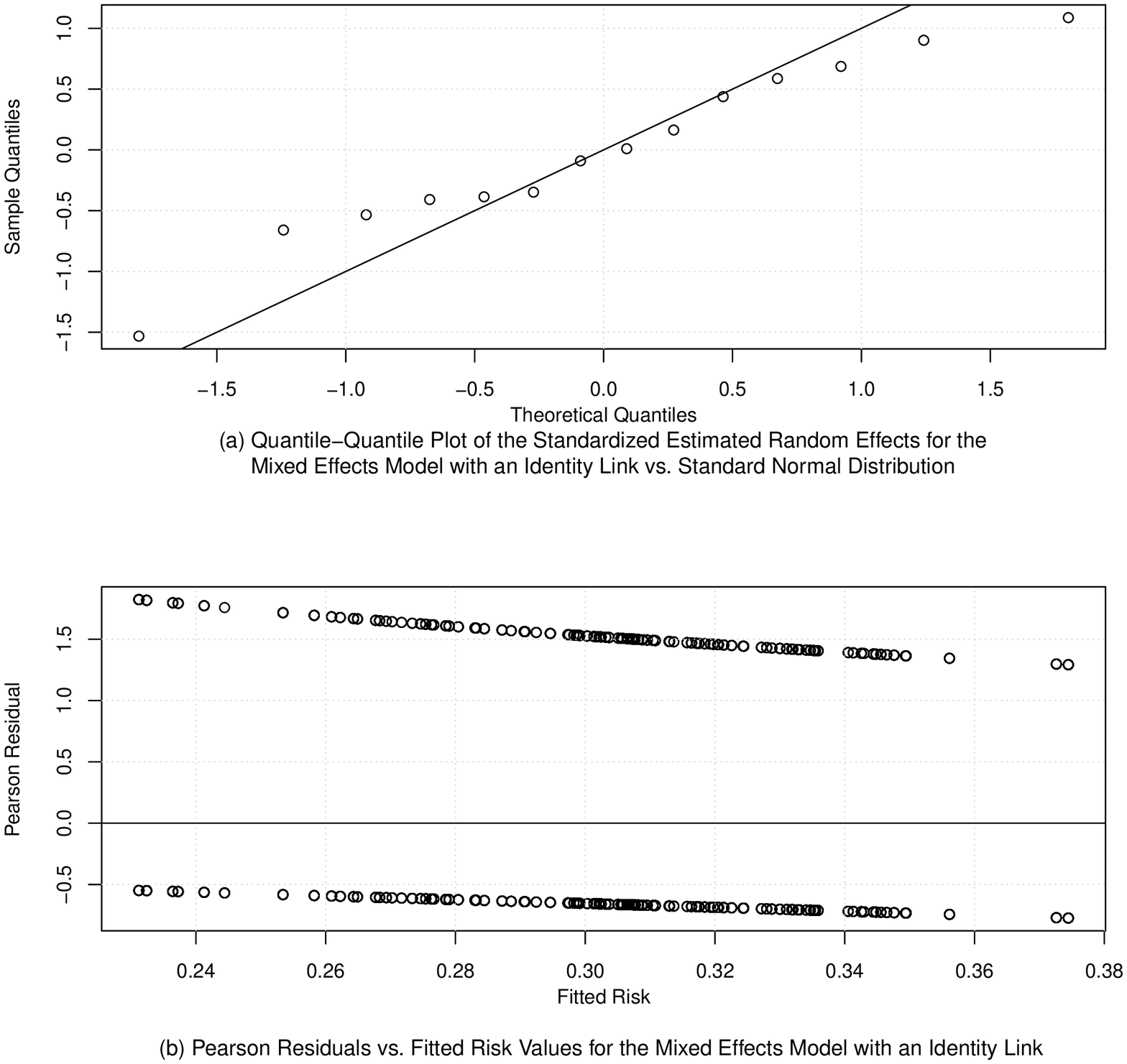}
\caption{Diagnostics for Mixed Effects Models for Tuberculosis SW-CRT}\label{FigC1}
\end{figure}
\end{center}
\end{document}